\documentclass[a4paper,11pt]{article}
\pdfoutput=1 

\usepackage[utf8]{inputenc}
\usepackage{jheppub}
\usepackage{natbib}

\usepackage{amssymb,amsmath,amsbsy,amsthm,mathtools}
\usepackage{braket}

\newtheorem{definition}{Definition}[section]
\newtheorem{proposition}{Proposition}[section]

\newtheorem{theorem}{Theorem}[section]

\theoremstyle{definition}
\newtheorem{example}{Example}[section]

\theoremstyle{definition}

\usepackage{todonotes}
\usepackage{graphicx,color}
\usepackage{multirow}
\usepackage{hyperref}
\usepackage{subcaption}
\newcommand{\mA}{\mathcal{A}}

\newcommand{\mH}{\mathcal{H}}

\newcommand{\mW}{\mathcal{W}}

\newcommand{\mbD}{\mathbb{D}}

\newcommand{\lp}{\left(}
\newcommand{\rp}{\right)}

\newcommand{\ep}{\varepsilon}

\newcommand{\p}{\partial}

\DeclareMathOperator{\arccosh}{arccosh}

\title{Phase transitions and uberholography of holographic pure-state geometries}

\author[a,b]{Ning Bao}
\author[a]{Keiichiro Furuya}
\author[a]{Jacob March}

\affiliation[a]{Department of Physics, Northeastern University, Boston, MA, 02115, USA}
\affiliation[b]{Computational Science Initiative, Brookhaven National Laboratory, Upton, NY 11973 USA}

\emailAdd{ningbao75@gmail.com}
\emailAdd{k.furuya@northeastern.edu}
\emailAdd{j.march@northeastern.edu}

\makeatletter
\gdef\@fpheader{}
\makeatother

\abstract{
We study the error-correcting properties of pure-state holographic geometries, in which mixed boundary subregions are replaced, via the surface/state correspondence, by the Ryu--Takayanagi (RT) geodesic bounding their entanglement wedges. 
In AdS$_3$/CFT$_2$ we derive a cross-ratio threshold relation $\eta'/\eta = e^{\Delta H/2}$ for the connected/disconnected transition of the entanglement wedge when two holes are punched in such a geometry. 
The quantity $\Delta H$ is sourced entirely by geodesics ending on RT boundaries. 
It shifts the standard two-interval threshold $\eta = 1/2$, and we classify when its sign is fixed by the pattern of hole endpoints. 
Turning to code properties, we show that the recursive hole-punching underlying uberholography cannot start within an RT-boundary, while an untouched asymptotic boundary can still fractalize, and we find numerically that in the configurations we study it does so with the universal fractal dimension $\alpha \approx 0.786$.
The resulting upper bounds on price and distance are nevertheless procedure dependent. 
In the configurations we study, punching holes on the asymptotic boundary while retaining the RT-boundary yields strictly tighter bounds than first tracing out the RT-boundary and then fractalizing.
}

\begin{document}
\maketitle

\section{Introduction}

The AdS/CFT correspondence \cite{Maldacena:1997re, Witten:1998qj, Gubser:1998bc} realizes a bulk gravitational theory as a quantum error correcting code defined on its holographic boundary \cite{Almheiri:2014lwa, Pastawski:2015qua}. 
In this dictionary the bulk degrees of freedom are the logical information that is redundantly encoded into the boundary so that they survive erasure of boundary regions. 
This perspective is aided by the Ryu-Takayanagi (RT) formula and its generalizations \cite{Ryu:2006bv, Ryu:2006ef, Hubeny:2007xt, Engelhardt:2014gca} which identify the boundary entanglement entropy with the area of a minimal bulk surface. 
This underlies entanglement-wedge reconstruction \cite{Dong:2016eik} and the RT formula itself has been derived from quantum error correction \cite{Harlow:2016vwg}. The code-theoretic content of a holographic geometry can be made quantitative through its \textit{price} and \textit{distance} \cite{Pastawski:2016qrs}. 
The price is defined as the smallest boundary region whose entanglement wedge contains a given bulk subregion while the distance is defined as the smallest boundary region whose erasure removes the ability to reconstruct the bulk subregion.
These quantities obey holographic analogues of the quantum Singleton bound and exhibit \textit{uberholography} \cite{Pastawski:2016qrs}: a bulk operator can be supported on a boundary region of fractal dimension $\alpha \approx 0.786$. 

A natural setting in which to probe these structures beyond the case of a pure global state is the \textit{surface/state correspondence} \cite{Miyaji:2015yva}, a proposed generalization of holography that assigns a quantum state to any codimension-two convex bulk surface and an RT-like entropy formula to its subregions. This makes it conjecturally possible to geometrically purify a mixed boundary state.
A boundary subregion in a mixed state can be replaced by a convex curve pushing the boundary inward.
By construction these surfaces can be pushed at most to the RT surface of the given subregion.
Iteratively pushing disjoint boundary subregions as far as the correspondence allows, produces an interesting class of pure-state geometries.
The boundary of the resulting geometry is a combination of the original asymptotic boundary and RT-boundary geodesics.
Such geometries have recently been used to study holographic entanglement distillation \cite{Bao:2023til} and they provide a concrete setup in which to ask how holographic code properties are modified by purification.
They are also closely tied to the entanglement wedge cross section (EWCS) and the entanglement of purification \cite{Takayanagi:2017knl, Nguyen:2017yqw} which sets the size of the minimal purifying system. 

In this paper we develop the phase structure and code properties of pure-state geometries in AdS$_3$/CFT$_2$.
Our results are organized as follows.

In section~\ref{sec:surface_state} we review the formulation of the surface/state correspondence and introduce the pure-state geometries that we will consider throughout the rest of the paper.

In section~\ref{sec:phase_dynamics} we study when the entanglement wedge of two disjoint holes punched in a pure-state geometry is connected. 
Each hole endpoint is an anchor that is either B-type (lying on the asymptotic boundary) or b-type (lying on an RT-boundary geodesic).
Theorem~\ref{thm:main} is a cross-ratio threshold relation $\eta'/\eta = e^{\Delta H/2}$ between two M\"obius invariants of the four anchors. 
This ratio is shifted away from the standard two-interval value $\eta=1/2$ \cite{Calabrese:2009ez, Calabrese:2010he, Headrick:2010zt, Hartman:2013mia, Faulkner:2013yia} by $\Delta H$ which is determined entirely by the geodesics connecting two b-type anchors.
We show when the sign of $\Delta H$ is fixed by the cyclic patterning of anchors and when it is not.
We organize these results by anchor pattern and illustrate the mechanism with several examples. 

In section~\ref{sec:price_and_distance} we define the price and distance of a pure-state geometry with respect to the surface/state RT formula, establish their equivalence for bulk points, and discuss the resulting holographic Singleton-type bound.

In section~\ref{sec:uberholography} we examine uberholography on pure-state geometries. 
We show that the recursive hole-punching underlying uberholography cannot be completed within an RT-boundary geodesic as an RT-boundary coincides with the bulk geodesic of the punched region.
Any asymptotic boundary that is left untouched, however, still fractalizes, and we find numerically that the fractal dimension $\alpha$ is unaltered.
In contrast, the resulting upper bounds on price and distance are dependent on whether the RT-boundaries are retained or traced out before the asymptotic boundary is fractalized. 
We explore an example of two such procedures and the implications of their differences.

We conclude in section~\ref{sec:discussion} with a discussion of open directions, including approximate bulk reconstruction and the relation between the Singleton gap and the failure of Haag duality. 

\section{Surface/state correspondence and pure-state geometries}\label{sec:surface_state}

In this section, we clean up our assumptions for what follows. Our work stands mainly on surface-state correspondence \cite{Miyaji:2015yva}.

The surface-state correspondence is proposed as a generalized holography \cite{Miyaji:2015yva}. In $d+2$ dimensional AdS, one can consider a codimension two convex surface $\Sigma$. The proposal is that the boundary theory lives on the surface. Hence there is a state associated with the surface. It comes with the principle such that
\begin{equation}
    \ket{\psi^\Sigma}_{SS} \in \mH_{tot} \leftrightarrow \Sigma\in M_{d+2} \; \text{(topologically trivial)} 
\end{equation}
and 
\begin{equation}
    \rho^\Sigma \in End(\mH_{tot}) \leftrightarrow \Sigma\in M_{d+2} \; \text{(topologically non-trivial)}.
\end{equation}
Furthermore, it conjectures that the surface state satisfies the RT formula, i.e.,
\begin{equation}\label{eq:RT}
    S(\rho^\Sigma_{\tilde{R}}) = \frac{Area(\gamma^\Sigma_{\tilde{R}})}{4G_N}.
\end{equation}


We call \textit{pure-state geometries} those that are topologically trivial and correspond to a pure state. In particular, we study the phase dynamics and code properties of entanglement wedges on a constant time slice, which are specific pure-state geometries, via the surface-state correspondence. 

Before we move on to the following section, we motivate the study of such holographic pure-states on the entanglement wedges and raise a few questions. The motivation is to understand how bulk logical information is encoded into the boundary of entanglement wedges to study holographic quantum error correction codes defined by a holographic mixed state and its purified state. 

For a holographic code on the asymptotic boundary of AdS$_{d+1}$, one can define an entanglement wedge corresponding to a boundary subregion. A boundary subregion can reconstruct the logical information in a bulk subregion if it is inside the entanglement wedge of the boundary subregion. 

For the holographic code on the entanglement wedge on which the state $\ket{\psi}_{SS}$ is assigned, entanglement wedges can be defined by the RT formula in \eqref{eq:RT}. A bulk geodesic of a segment of an RT surface stays on the surface, see figure \ref{fig:no_supadd}. Thus, it cannot reconstruct bulk logical information that does not intersect with the segment. However, together with the RT surface and the boundary subregion, especially when the entanglement wedge is the connected phase, provides much more access to the bulk logical information. That is, the wedge\footnote{The union of the segments of RT surfaces and the EWCS was called a kinked RT surface in \cite{Hayden:2021gno}.} bounded by the boundary subregion, segments of RT surfaces and the EWCS is larger than the entanglement wedge of a boundary subregion. Thus, understanding the role of RT surfaces is essential to explore the holographic encoding of bulk logical information into a pure state on the boundary of the entanglement wedges. Such a situation arises in the purification of a holographic mixed state. 

Several purifications of a holographic mixed state have been studied, such as a canonical purification \cite{Dutta:2019gen,Sorce:2025usc} and the optimal purification via entanglement of purification \cite{Terhal:2002riz,Takayanagi:2017knl,Caputa2019hologrpahicEoP,Nguyen:2017yqw}\footnote{We refer the reader to the cited works and references therein, while acknowledging that the literature on this topic is much broader than can be fully covered here.}. In \cite{Takayanagi:2017knl}, the entanglement of purification $E_p(A:B)$ of a reduced state on the boundary, for instance, $\rho_{AB}$, has been conjectured to be equivalent to the entanglement wedge cross section $EWCS(A:B)$. The paper \cite{Bao:2018zab} proposed a holographic minimal purification, where, for example, a mixed state $\rho_{AB}$ is purified with a purifying system of minimal size, optimizing $E_p(A:B)$. It also argued based on the bit-threads approach \cite{Du:2019emy,Bao:2018zab} that, for the conjecture $E_p=E_W$ to hold, $A^*$ and $B^*$ must reside on the segments of the RT surface of the connected entanglement wedge of $AB$. This suggests that the holographic purification via $E_p$ can be done geometrically.

We now ask the following question. Suppose we consider a geometric minimal purification for a mixed state $\rho_{AB}$ and denote by $\ket{\psi_{AB(AB)^*}}_{E_p}$ the pure-state state obtained by the purification procedure. Does the state $\ket{\psi}_{E_p}$ satisfy the RT formula for any proper subregion $\tilde{R}$? We can also ask if there is a state $\ket{\psi}_{SS}$ on the boundary of an entanglement wedge induced from the surface-state correspondence, which can also geometrically optimize the corresponding entanglement of purification. In this paper, we do not resolve these questions.

\section{Phase dynamics on a pure-state geometry}\label{sec:phase_dynamics}

We classify bipartite phase transitions in pure-state geometries in AdS${}_3$/CFT${}_2$, comparing to the standard AdS/CFT results seen in \cite{Calabrese:2009ez, Calabrese:2010he, Headrick:2010zt, Hartman:2013mia}. 
Throughout we work on a static time slice of pure AdS$_3$, so extremal surfaces are geodesics of the hyperbolic disk and the RT rather than HRT prescription applies. 
 
\subsection{Cross-ratio threshold relation}
We work in the Poincar\'e disk $\mathbb{D} = \{z \in \mathbb{C} : |z| < 1\}$ with the hyperbolic metric
\begin{equation}
    ds^2 = \frac{4|dz|^2}{(1 - |z|^2)^2}.
\label{eq:metric}
\end{equation}
The induced hyperbolic distance between two points $z_1, z_2 \in \mathbb{D}$ is
\begin{equation}
    d(z_1, z_2) = \arccosh\lp1 + 2\rho_{12}\rp, 
    \qquad 
    \rho_{ij} := \frac{|z_i - z_j|^2}{(1 - |z_i|^2)(1 - |z_j|^2)},
    \label{eq:hypdist}
\end{equation}
and bulk geodesics are arcs of circles orthogonal to $\partial\mathbb{D}$ (diameters included). The AdS$_3$ radius is set to unity, so all lengths and the cutoff $\varepsilon$ below are dimensionless. 
 
We start with a given pure-state geometry where we have pushed disjoint open boundary intervals inward to their RT geodesics.
The resulting subsystem boundary $\Sigma$ is a closed curve in $\bar{\mathbb{D}}$ made of original-boundary arcs on $\partial\mathbb{D}$ and RT-boundary geodesics, meeting at corners. 
A punched out hole is an open connected subset of $\Sigma$ that is traced out.
Two disjoint holes produce four anchors $z_1, \ldots, z_4$ in cyclic counter-clockwise order around $\Sigma$, with $\{z_1, z_2\}$ and $\{z_3, z_4\}$ the hole endpoints. 
The remaining subsystem boundary consists of two un-punched arcs (from $z_2$ to $z_3$ and from $z_4$ to $z_1$). 
Each anchor is either \textit{B-type} (on the original boundary, $|z_i| = 1$) or \textit{b-type} (on an RT-boundary, $|z_i| < 1$) and we will refer to geodesics that end on two B-type (b-type) points as \textit{BB} (\textit{bb}) and for mixed anchors \textit{Bb}.
See figure~\ref{fig:setup}.
Corners, where RT-boundary geodesics meet the asymptotic boundary or each other, are isolated points on $\Sigma$, and we assume throughout that no anchor sits exactly at a corner, so every anchor has a definite type.
 
\begin{figure}[ht]
    \centering
    \includegraphics[width=0.5\linewidth]{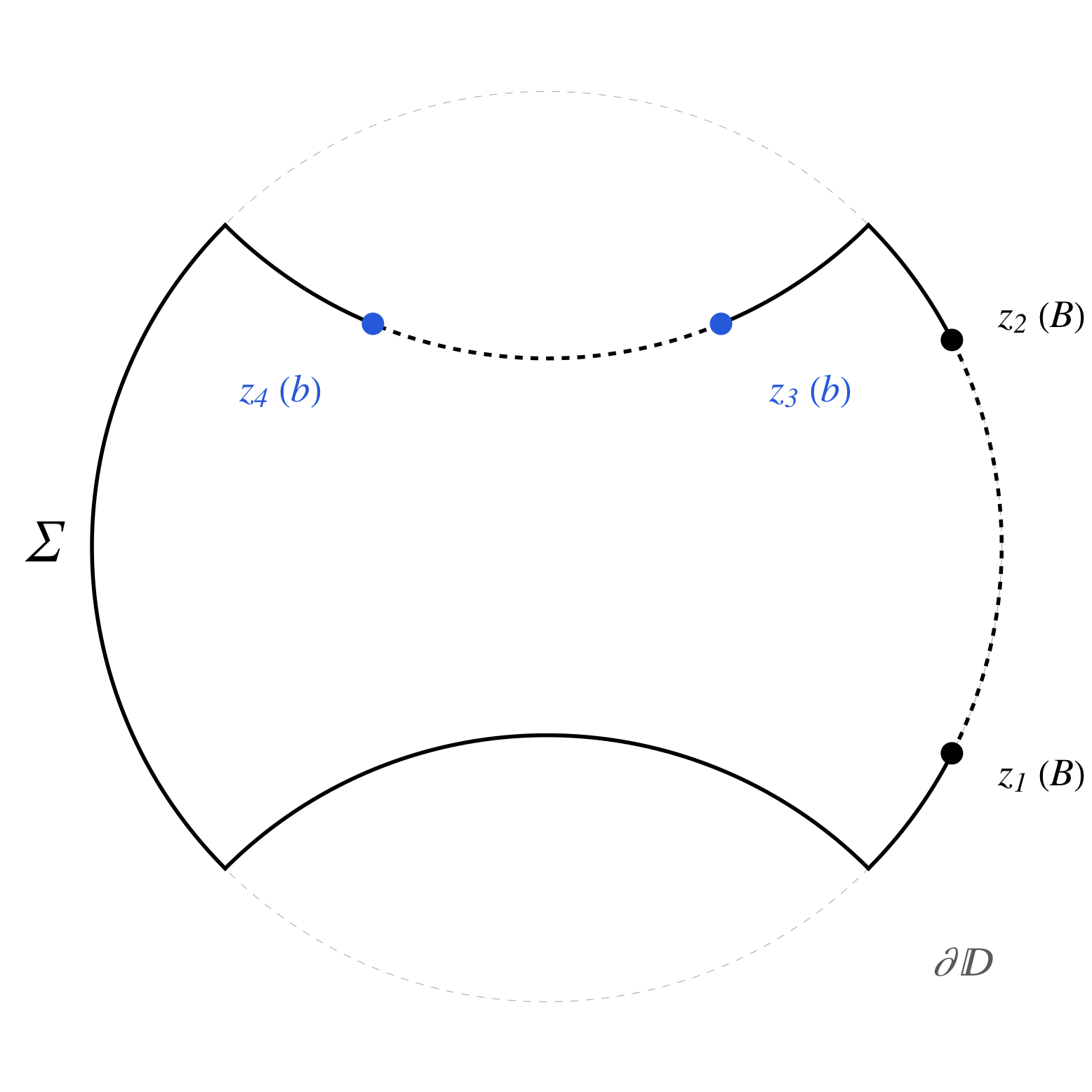}
    \caption{A pure-state geometry $\Sigma$ in $\mathbb{D}$ after a bipartite partial trace. 
    Anchors are labeled in cyclic counter-clockwise order. B-type anchors lie on $\partial\mathbb{D}$, b-type anchors lie on RT-boundary geodesics.
    Dashed segments indicate traced-out regions where holes have been punched.
    }
    \label{fig:setup}
\end{figure}
 
In order to regularize our boundary anchored geodesics, we push each boundary anchor inward to $|z| = 1-\ep$ with $\ep > 0$ small. 
Then, to leading order, we have

\begin{equation}
    1 - |z|^2 \;\longmapsto\; 2\varepsilon \qquad (z \in \partial\mathbb{D}).
    \label{eq:cutoff}
\end{equation}
For each pair of hole endpoints $(z_i, z_j)$ let $n_B(z_i, z_j) \in \{0, 1, 2\}$ count its number of B-type anchors.
Under \eqref{eq:cutoff}, the hyperbolic distance $d(z_i, z_j)$ is finite for $n_B = 0$ and diverges as $-n_B \log\ep$ for $n_B\ge 1$.
Its asymptote is $d(z_i, z_j) \to \log (4\rho_{ij})$ as $\rho_{ij}\to\infty$.
 
By the RT prescription, the geodesics replacing the holes are homologous to the un-punched arcs \cite{Ryu:2006bv, Ryu:2006ef, Headrick:2007km}. The bulk wedge bounded by un-punched arcs and RT geodesics is connected or disconnected according to the geodesic pairing:
\begin{align}
    L_{\mathrm{conn}} &= d(z_1, z_2) + d(z_3, z_4), \\
    L_{\mathrm{disc}} &= d(z_1, z_4) + d(z_2, z_3).
\end{align}
Under $L_{\mathrm{conn}}$, arcs and geodesics assemble into one closed loop bounding a single wedge. 
Under $L_{\mathrm{disc}}$, each un-punched arc is capped by its own geodesic, bounding two disjoint wedges. 
See figure~\ref{fig:channels}. 
These two non-crossing pairings exhaust the candidate minimal surfaces. 
The third possible pairing of $(z_1, z_3)$ and $(z_2, z_4)$ has crossing geodesics.
Rerouting the two surfaces at the crossing produces kinked curves in the non-crossing homology classes whose lengths are strictly reduced by smoothing, so the crossing pairing is never minimal \cite{Headrick:2007km}.
Moreover, each candidate geodesic lies within the pure-state geometry, since the surface/state correspondence leaves the bulk metric unaltered and $\Sigma$ is geodesically convex, as required by the surface/state correspondence \cite{Miyaji:2015yva}.
Each $L$ may diverge individually, but every anchor appears in one pair of each channel, so the $-n_B \log\ep$ contributions cancel pairwise for $L_\text{disc} - L_\text{conn}$, which is finite as $\ep \to 0$.
In AdS$_{3}$ the extremal surfaces are geodesics, so by \eqref{eq:RT} the entanglement entropy of the un-punched arcs is 
$\min(L_\text{conn}, L_\text{disc})/4G_N$. 
The RT prescription selects this minimum, with the phase transition at $L_\text{conn} = L_\text{disc}$ \cite{Ryu:2006bv, Ryu:2006ef, Headrick:2010zt, Hartman:2013mia}.
 
\begin{figure}[ht]
    \centering
    \includegraphics[width=1.\linewidth]{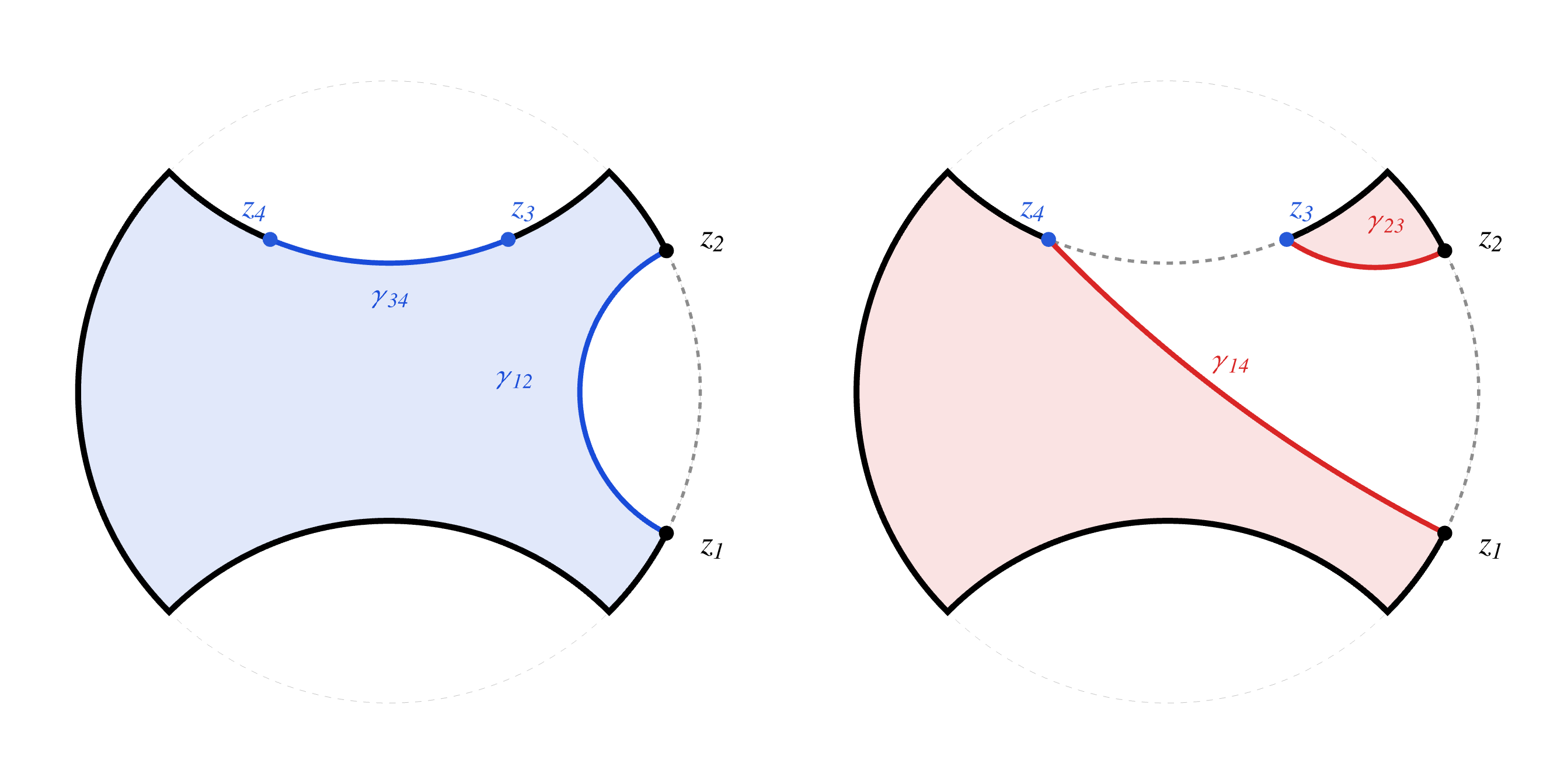}
    \caption{
    The two non-crossing pairings and their bulk wedge topologies. 
    Left panel: connected channel $L_{\mathrm{conn}} = d(z_1, z_2) + d(z_3, z_4)$, with both geodesics bridging the un-punched arcs.
    In this panel geodesic $\gamma_{12}$ is BB while $\gamma_{34}$ is bb. 
    Right panel: disconnected channel $L_{\mathrm{disc}} = d(z_1, z_4) + d(z_2, z_3)$, with each geodesic capping one un-punched arc.
    Here both geodesics are Bb.
    }
    \label{fig:channels}
\end{figure}
 
Define the cross-ratios
\begin{equation}
    \eta := \frac{|z_1 - z_2| \, |z_3 - z_4|}{|z_1 - z_3| \, |z_2 - z_4|}, \qquad
    \eta' := \frac{|z_1 - z_4| \, |z_2 - z_3|}{|z_1 - z_3| \, |z_2 - z_4|}.
    \label{eq:eta}
\end{equation}
These are M\"obius-invariant under the bulk isometry group $PSU(1,1).$\footnote{
$PSU(1,1)$ acts on $\bar{\mathbb{D}}$ as disk-preserving M\"obius transformations and on $\p \mbD$ as the boundary CFT global conformal group. 
For four cyclically-ordered boundary points it has just enough parameters to gauge away all but $\eta$.
Ptolemy's inequality  $|z_1 - z_3|\,|z_2 - z_4| \leq |z_1 - z_2|\,|z_3 - z_4| + |z_1 - z_4|\,|z_2 - z_3|$ is saturated, fixing $\eta' = 1-\eta$ \cite{Calabrese:2009ez, Headrick:2010zt}.
Interior anchors break saturation in general, leaving $\eta'$ no longer determined by $\eta$.
For \textit{concyclic} interior anchors saturation is recovered ($\eta + \eta' = 1$) as used in examples~\ref{ex:HH} and \ref{ex:rect}} 
 
It will also be useful to define a quantity with distinct quantitative behavior for different endpoint configurations. 
For this we introduce $h(\rho)$, the difference between the hyperbolic distance and its large-$\rho$ asymptote:
\begin{equation}
    h(\rho_{ij}) :=d(z_i, z_j) - \log (4\rho_{ij}).
    \label{eq:hfun}
\end{equation}
This function is shown in figure~\ref{fig:hplot}. 
Note that $h(\rho_{ij})$ is strictly positive and strictly decreasing on $\rho_{ij} \in(0, \infty)$ and  $h(\rho_{ij}) \to +\infty$ as $\rho_{ij} \to 0^+$\footnote{Explicitly, $h'(\rho)=(\rho^2 + \rho)^{-1/2} - \rho^{-1} < 0$ on $\rho\in (0, \infty)$ and $h(\rho)\to0$ as $\rho \to \infty$ so monotonicity implies positivity.}. 
Note that in the limit $\ep \to 0$ we have that $\rho_{ij}$ diverges if and only if $n_B(z_i, z_j) \ge 1$ (i.e. $\gamma_{ij}$ is not bb). Since $h(\rho) \to 0$ as $\rho \to \infty$, the contribution of a non-bb pair vanishes in this limit, thus
\begin{equation}
    h(\rho_{ij}) = 
    \begin{cases} 
        \arccosh(1 + 2\rho_{ij}) - \log(4\rho_{ij}) & n_B(z_i, z_j) = 0, \\ 
        0 & n_B(z_i, z_j) \geq 1, 
    \end{cases}
    \label{eq:hpair}
\end{equation}
 
\begin{figure}[ht]
    \centering
    \includegraphics[width=0.6\linewidth]{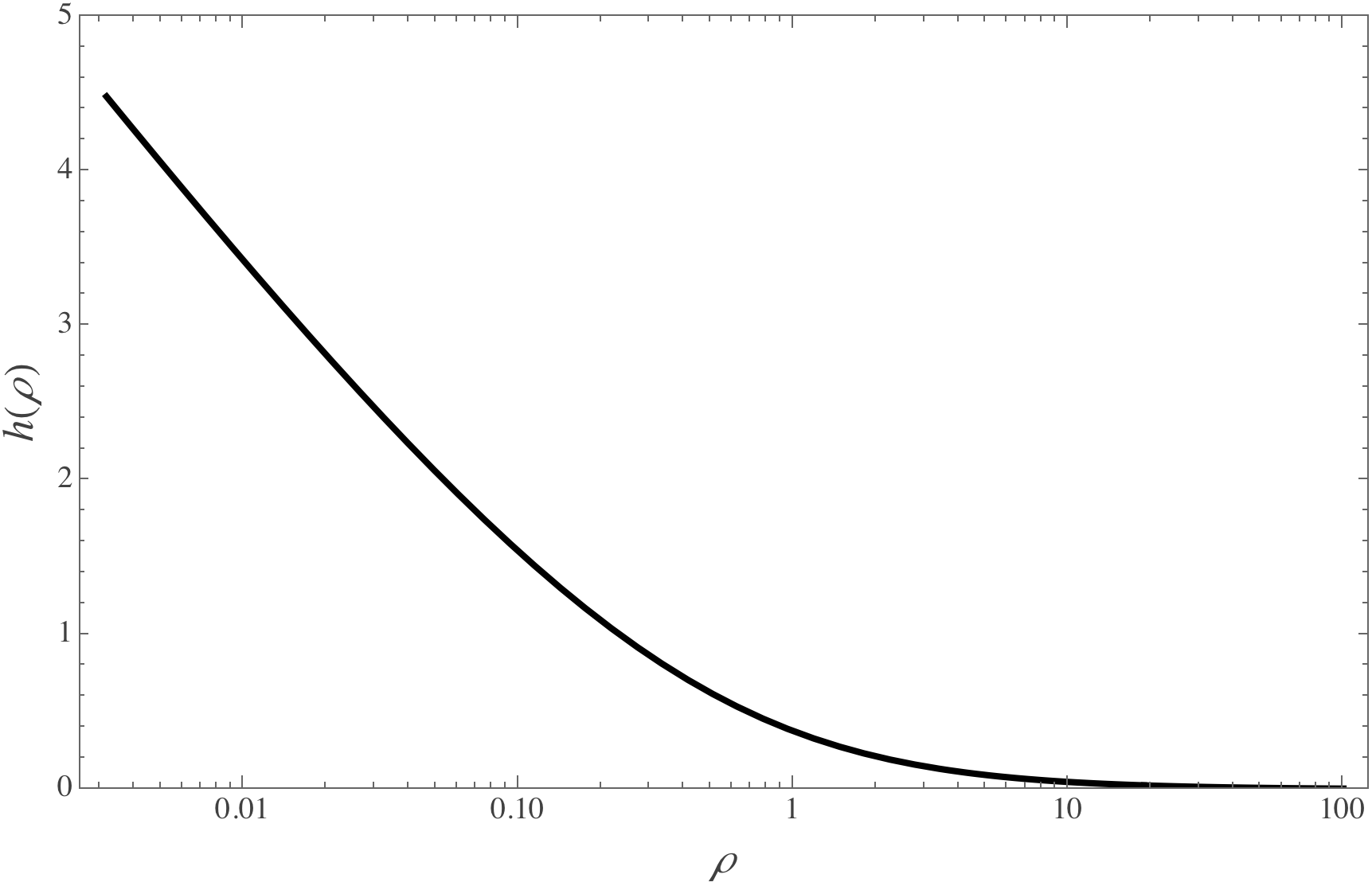}
    \caption{The function $h(\rho) = \arccosh(1 + 2\rho) - \log(4 \rho)$, the difference between the hyperbolic distance and its large-$\rho$ asymptote. Strictly positive and strictly decreasing on $(0, \infty)$, with asymptote $h(\rho) \to 0$ as $\rho \to \infty$ and divergence as $\rho \to 0^+$.}
    \label{fig:hplot}
\end{figure}
 Let $H := \sum_{(i,j)} h(\rho_{ij})$ for a given connected or disconnected channel such that $H_\text{conn} = h(\rho_{12}) + h(\rho_{34})$ and $H_\text{disc} = h(\rho_{14}) + h(\rho_{23})$.

\begin{theorem}[Cross-ratio threshold relation]
    \label{thm:main}
    Let two disjoint holes be punched on a pure-state geometry, producing four distinct anchors, none at a corner of $\Sigma$. 
    Define $\Delta H :=  H_{\text{conn}} - H_\text{disc}$, computed via~\eqref{eq:hpair}.
    Then in the limit $\ep\to 0$ the entanglement wedge of the un-punched boundary regions is connected if and only if 
    \begin{equation}
            \frac{\eta'}{\eta} > e^{\Delta H/2},
        \label{eq:main}
    \end{equation}
    and disconnected if the inequality is reversed, with the phase transition at equality. 
    In the all-boundary case (i.e. all anchors are B-type), we have $\eta'=1-\eta$ and $\Delta H=0$, recovering the standard threshold $\eta < 1/2$ for the connected phase \cite{Headrick:2010zt, Hartman:2013mia}.
\end{theorem}

\noindent\textit{Proof.}  We work first at finite $\ep$, with the B-type anchors at their regulated positions. By \eqref{eq:hypdist} and \eqref{eq:hfun}, the total channel lengths split as 
\begin{equation}
  L_\text{conn} =  \log (4\rho_{12}) + \log (4\rho_{34}) + H_\text{conn},
  \qquad 
  L_\text{disc} =  \log (4\rho_{14}) + \log (4\rho_{23}) + H_\text{disc}.
\end{equation}
Subtracting gives
\begin{equation}
    L_{\mathrm{disc}} - L_{\mathrm{conn}} = \log\frac{\rho_{14}\rho_{23}}{\rho_{12}\rho_{34}} + (H_{\mathrm{disc}} - H_{\mathrm{conn}}).
    \label{eq:cancel}
\end{equation}
 Writing $f(z):= 1-|z|^2$ (for a regulated B-type anchor this equals $2\ep$ to leading order as $\ep\to0$), we have $\rho_{ij} = |z_i-z_j|^2/(f(z_i)f(z_j))$. Each $f(z_k)$ appears once in the numerator and once in the denominator of the $\rho$-ratio in \eqref{eq:cancel}, so all cancel exactly: 
\begin{equation}
    \label{eq:cr}
    \frac{\rho_{14}\,\rho_{23}}{\rho_{12}\,\rho_{34}} 
    =
    \frac{|z_1 - z_4|^2 |z_2 - z_3|^2}{|z_1 - z_2|^2 |z_3 - z_4|^2} 
    =
    \left(\frac{\eta'}{\eta}\right)^{2},
\end{equation}
where the common factor $|z_1-z_3|\,|z_2-z_4|$ of $\eta, \eta'$ cancels and the cross-ratios are evaluated at the regulated positions. 
Putting this back into \eqref{eq:cancel} gives 
\begin{equation}
    L_{\mathrm{disc}} - L_{\mathrm{conn}} = 2 \log\left(\frac{\eta'}{\eta}\right) - \Delta H.
\label{eq:Ldiff}
\end{equation}
Every term in \eqref{eq:Ldiff} has a finite limit as $\ep\to0$. 
The cross-ratios tend to their values at the unregulated anchor positions, while by \eqref{eq:hpair} only bb pairs survive in $\Delta H$.
The statement of the theorem is therefore well posed as both sides of~\eqref{eq:main} are $\ep$-independent.
In order to have a connected bulk entanglement wedge we need $L_\text{conn} < L_\text{disc}$ (reversed for disconnected) which yields \eqref{eq:main} with equality $L_\text{conn} = L_\text{disc}$ corresponding to $\eta'/\eta = e^{\Delta H/2}$.
 
\subsection{Examples and classification}\label{sec:examples}

Here we walk through some examples and discuss the patterns of phase transition behavior that emerge. 

\begin{figure}[ht]
    \centering
    \newcommand{\panelwidth}{0.22\linewidth}
    \begin{subfigure}[t]{\panelwidth}\centering
      \includegraphics[width=\linewidth]{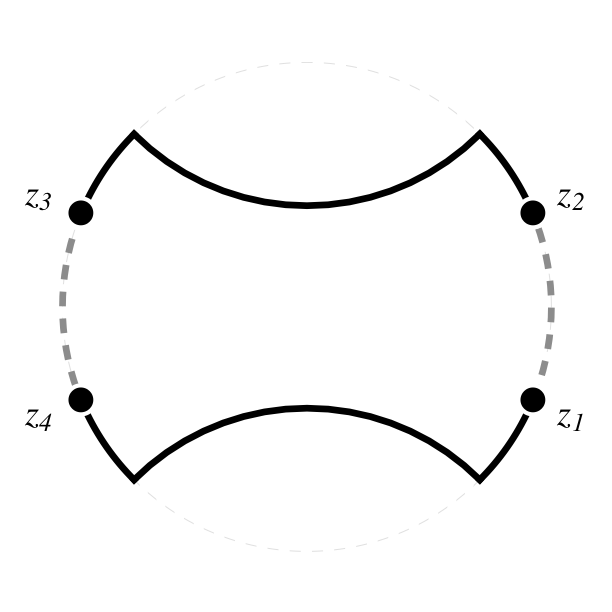}
      \subcaption{B--B--B--B}\label{fig:pat_a}
    \end{subfigure}\hfill
    \begin{subfigure}[t]{\panelwidth}\centering
      \includegraphics[width=\linewidth]{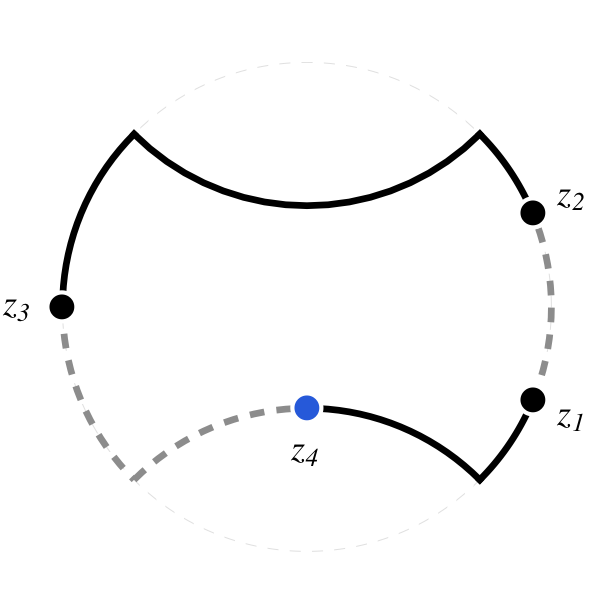}
      \subcaption{3B + 1b}\label{fig:pat_b}
    \end{subfigure}\hfill
    \begin{subfigure}[t]{\panelwidth}\centering
      \includegraphics[width=\linewidth]{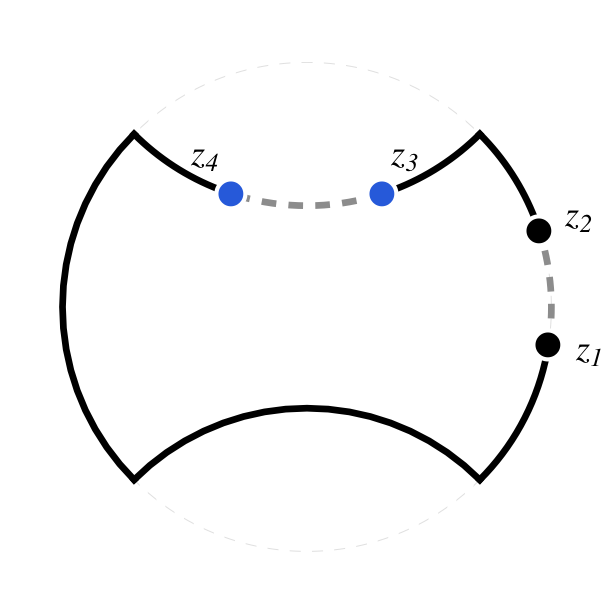}
      \subcaption{B--B--b--b}\label{fig:pat_c}
    \end{subfigure}\hfill
    \begin{subfigure}[t]{\panelwidth}\centering
      \includegraphics[width=\linewidth]{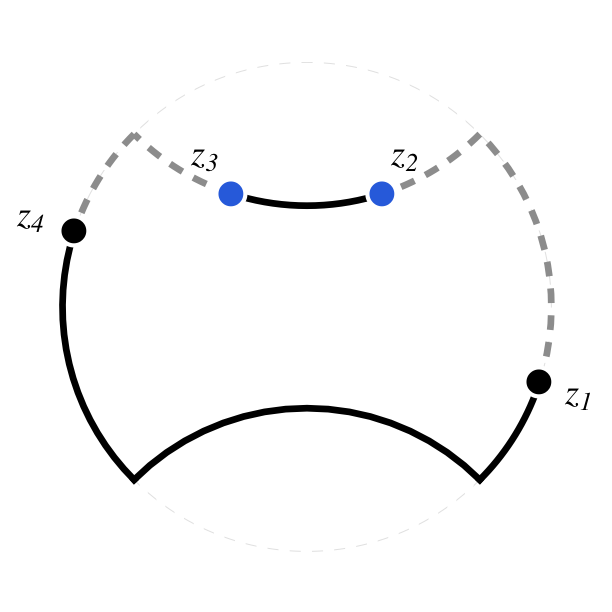}
      \subcaption{B--b--b--B}\label{fig:pat_d}
    \end{subfigure}
     
    \vspace{1em}
     
    \begin{subfigure}[t]{\panelwidth}\centering
      \includegraphics[width=\linewidth]{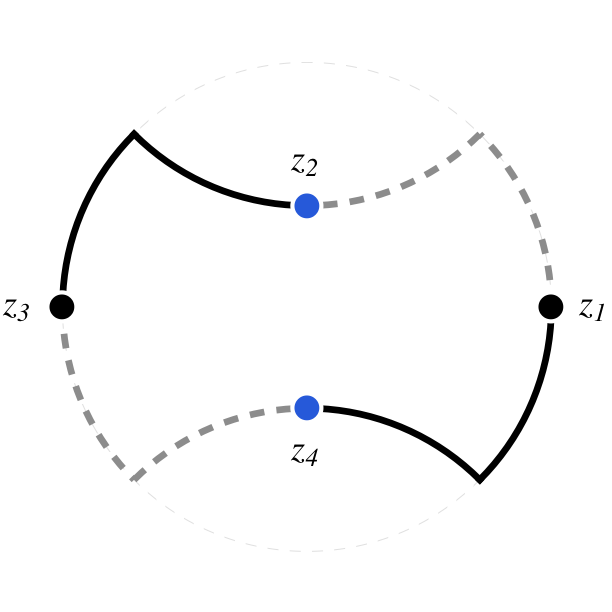}
      \subcaption{B--b--B--b}\label{fig:pat_e}
    \end{subfigure}\hfill
    \begin{subfigure}[t]{\panelwidth}\centering
      \includegraphics[width=\linewidth]{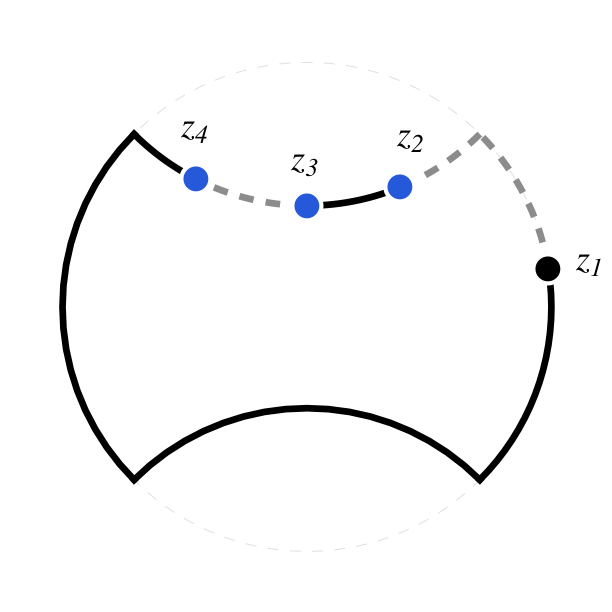}
      \subcaption{3b + 1B}\label{fig:pat_f}
    \end{subfigure}\hfill
    \begin{subfigure}[t]{\panelwidth}\centering
      \includegraphics[width=\linewidth]{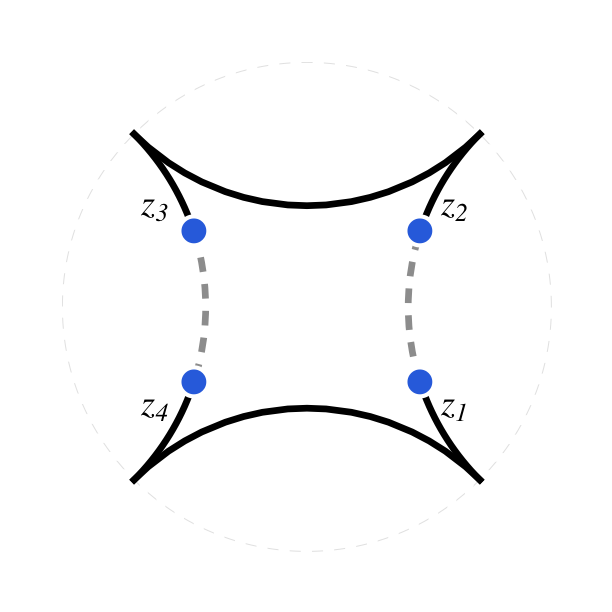}
      \subcaption{b--b--b--b (sym.)}\label{fig:pat_g}
    \end{subfigure}\hfill
    \begin{subfigure}[t]{\panelwidth}\centering
      \includegraphics[width=\linewidth]{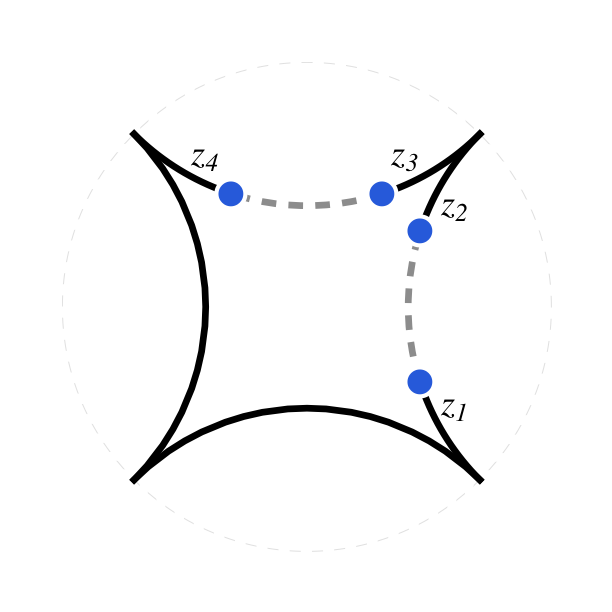}
      \subcaption{b--b--b--b (asym.)}\label{fig:pat_h}
    \end{subfigure}
     
    \caption{Cyclic anchor-type patterns at four marks on $\Sigma$. Each loop shows the cyclic anchor sequence and the hole-pair assignment $\{z_1,z_2\}\cup\{z_3,z_4\}$. The sign of $\Delta H$ for each pattern is collected in table~\ref{tab:patterns}.}
    \label{fig:patterns}
\end{figure}

\begin{example}[All-boundary anchors]
    \label{ex:HH}
    With four boundary anchors, all four pairs are BB so $H_{\mathrm{conn}} = H_{\mathrm{disc}} = 0$ and $\eta + \eta' = 1$ (anchors on $\partial\mathbb{D}$ in cyclic order saturate Ptolemy). 
    The threshold \eqref{eq:main} reduces to $\eta'/\eta = 1$, placing the transition at $\eta = 1/2$.
\end{example}
 
\begin{example}[Inscribed-rectangle family]
    \label{ex:rect}
    Take four b-type anchors M\"obius-equivalent to vertices of a rectangle inscribed in a circle in $\bar{\mathbb{D}}$. 
    The vertices are concyclic (i.e. sit on a shared circle), so $\eta + \eta' = 1$.
    By disk-preserving M\"obius conjugation we center the circumscribing circle at the origin, making the rectangle's perpendicular bisectors hyperbolic isometries. 
    Under these reflections $\rho_{12} = \rho_{34}$ and $\rho_{14} = \rho_{23}$, so $H_{\mathrm{conn}} = 2 h(\rho_{12})$ and $H_{\mathrm{disc}} = 2 h(\rho_{14})$. 
    The transition condition $\eta'/\eta = e^{\Delta H/2}$ with $\Delta H =2( h(\rho_{12}) - h(\rho_{14}))$, combined with $\eta + \eta' = 1$, gives $\eta=1/2$ exactly when $\rho_{12} = \rho_{14}$, i.e. at the square.
    This family is realized, for instance, on the symmetric entanglement arena of example~\ref{ex:arena} below with one anchor on each of the four RT-boundary geodesics and holes straddling two opposite corners, as in the symmetric b--b--b--b placement of figure~\ref{fig:patterns}(g).
\end{example}
 
\begin{example}[B--B--b--b]
    \label{ex:bbpp}
    With cyclic pattern B--B--b--b, the connected pair $(z_3, z_4)$ is bb and $(z_1, z_2)$ is BB. 
    Both disconnected pairs are Bb, so $H_\text{conn} = h(\rho_{34})$ and $H_\text{disc} = 0$, giving $\Delta H = h(\rho_{34}) > 0$. 
    The transition threshold $\eta'/\eta = e^{\Delta H/2}$ requires $\eta < \eta'$ at transition. 
\end{example}
 
\begin{example}[B--b--b--B]
    \label{ex:bppb}
    With cyclic pattern B--b--b--B ($z_1, z_4$ B-type, $z_2, z_3$ b-type), the disconnected pair $(z_2, z_3)$ is bb and $(z_1, z_4)$ is BB. 
     See figure~\ref{fig:geom_examples}(a).
    Both connected pairs are Bb, so $H_{\text{conn}} = 0$ and $H_{\text{disc}} = h(\rho_{23})$, giving $\Delta H = - h(\rho_{23})< 0$.
    The transition threshold $\eta'/\eta = e^{\Delta H/2} < 1$ requires $\eta > \eta'$ at transition.
\end{example}
 
\begin{example}[Entanglement arena: all-RT-boundary with adjacent geodesic holes]
    \label{ex:arena}
    The initial geometry here (see figure~\ref{fig:geom_examples}(b)) is what results from completing the RT purification on all intervals \cite{Miyaji:2015yva, Bao:2023til}, and we refer to it as the \textit{entanglement arena}.
    We now punch holes on two adjacent RT-boundary geodesics sharing a corner, taking the holes mirror-symmetric about the diagonal through the shared corner as in figure~\ref{fig:geom_examples}, and label so $z_2$ and $z_3$ lie closest to the corner.
    Since $z_1, z_2$ (and likewise $z_3, z_4$) lie on a single RT-boundary geodesic, the connected-channel geodesics coincide with the punched holes themselves. 
    In the connected phase the wedge remains the full arena, while in the disconnected phase it retreats to two wedges capping the un-punched arcs.
    All four pairs are bb, with
    \begin{equation}
        H_{\mathrm{conn}} = h(\rho_{12}) + h(\rho_{34}), \qquad H_{\mathrm{disc}} = h(\rho_{14}) + h(\rho_{23}).
    \end{equation}
    Reflection through the shared-corner diagonal gives $\rho_{12} = \rho_{34}$,
    so $\Delta H$ decomposes as
    \begin{equation}
        \Delta H = 2h(\rho_{12}) - h(\rho_{14}) - h(\rho_{23})
        = [\,h(\rho_{12}) - h(\rho_{14})\,]
        \;-\; [\,h(\rho_{23}) - h(\rho_{12})\,].
    \end{equation}
    Unlike example~\ref{ex:bbpp} ($\Delta H > 0$) and example~\ref{ex:bppb} ($\Delta H < 0$), the sign of $\Delta H$ for the arena is not fixed by the anchor pattern. 
\end{example}
 
\begin{figure}[ht]
    \centering
    \includegraphics[width=\linewidth]{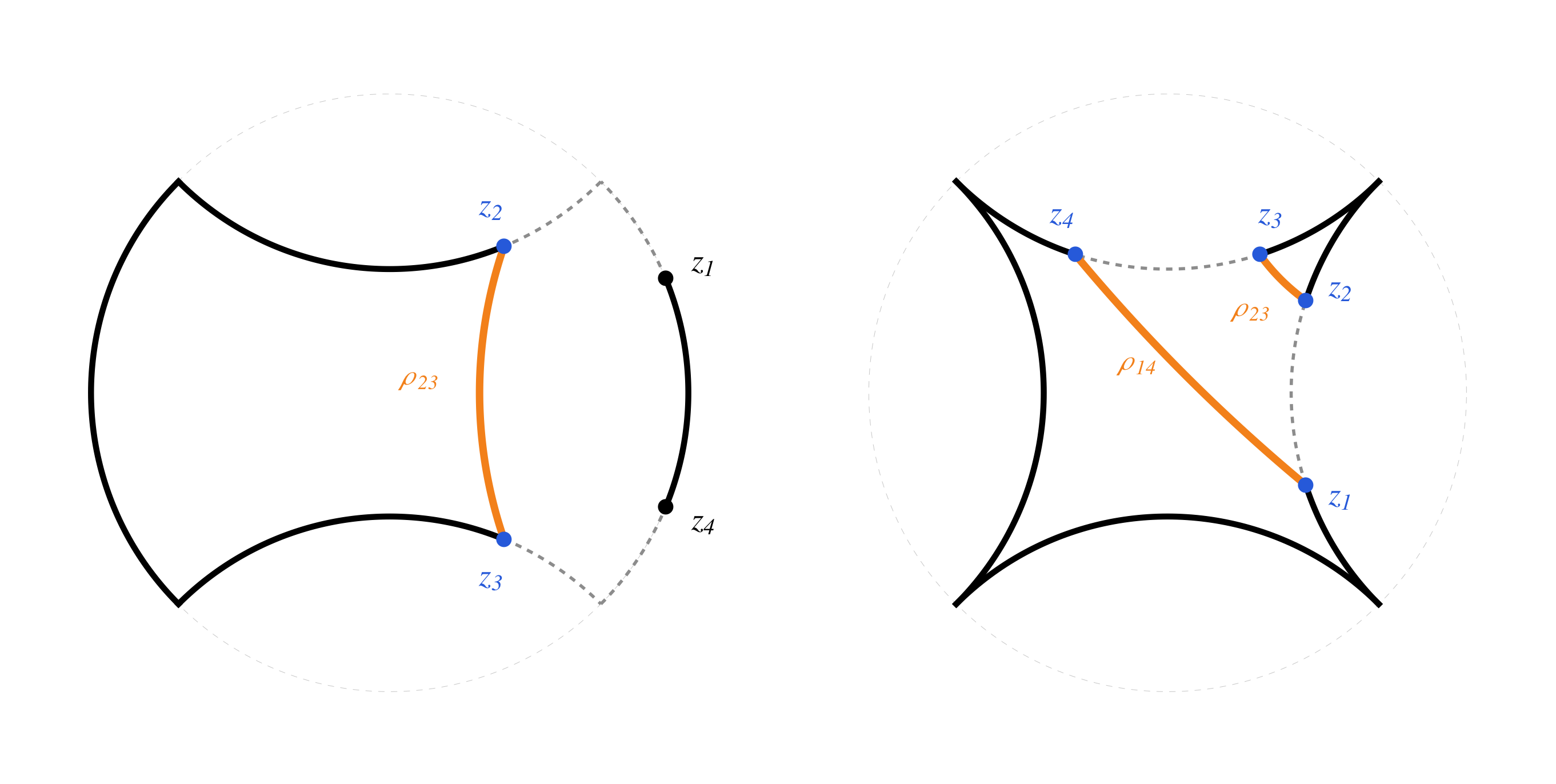}
    \caption{Two examples with bb geodesics shown in orange.
    Left: B--b--b--B, with holes straddling corners. 
    Right: the entanglement arena, with holes on adjacent RT-boundary geodesics sharing a corner. 
    In both configurations as drawn, the disconnected channel inherits the short bb geodesic $\gamma_{23}$, supplying the dominant contribution to $\Delta H$.}
    \label{fig:geom_examples}
\end{figure}

We now summarize the classification. 
Theorem~\ref{thm:main} expresses the transition as a single relation $\eta'/\eta = e^{\Delta H/2}$ between M\"obius invariants of the configuration, with $\Delta H = H_\text{conn} - H_\text{disc}$. 
The sign of $\Delta H$ sets the direction of the shift, as $\Delta H> 0$ requires $\eta'>\eta$ at transition, $\Delta H<0$ requires $\eta'<\eta$, and $\Delta H=0$ recovers the symmetric condition $\eta =\eta'$. 
When $\Delta H = 0$, the condition $\eta = \eta'$ combined with $\eta + \eta'=1$ (i.e., the anchors lie on a shared circle) places the transition at $\eta=1/2$. 
Note that neither concyclicity nor $\Delta H$ alone are sufficient for this critical $\eta$ as concyclic anchors can arise in patterns with $\Delta H \ne 0$ and vice versa.

Because every non-bb pair contributes $0$ under \eqref{eq:hpair}, only bb pairs enter
\begin{equation}
    \Delta H = \big(h(\rho_{12}) + h(\rho_{34})\big) - \big(h(\rho_{14}) + h(\rho_{23})\big),
\end{equation}
each through a strictly positive $h(\rho_{ij})$.
This is a difference of two nonnegative terms from which non-bb pairs drop out.
The anchor pattern fixes which sums vanish, while the $\rho_{ij}$ set only the magnitude of the nonzero ones.
Hence the sign of $\Delta H$ is fixed by the pattern whenever bb pairs lie in at most one channel (i.e. $\Delta H=0$ for no bb pair, $\Delta H>0$ for bb only in the connected channel, $\Delta H<0$ for bb only in the disconnected channel).
When both channels contain bb pairs, $\Delta H$ subtracts two positive sums and its sign is set by the $\rho_{ij}$.
Table~\ref{tab:patterns} collects the resulting signs of $\Delta H$ and the channel structure is illustrated in figure~\ref{fig:patterns}.
\begin{table}[ht]
\centering
    \begin{tabular}{llc}
        \textbf{Cyclic pattern} & \textbf{Sign of} $\Delta H$ & \textbf{Fig.~\ref{fig:patterns}} \\ \hline
        B--B--B--B & $\Delta H=0$ & (a) \\
        3B + 1b & $\Delta H=0$ & (b) \\
        B--B--b--b & $\Delta H>0$ & (c) \\
        B--b--b--B  & $\Delta H < 0$ & (d) \\
        B--b--B--b & $\Delta H=0$ & (e) \\
        3b + 1B & depends on $\rho$'s & (f) \\
        b--b--b--b (symmetric) & $\Delta H=0$ & (g) \\
        b--b--b--b (asymmetric) & depends on $\rho$'s & (h) \\
    \end{tabular}
    \caption{The sign of $\Delta H$ for the cyclic anchor-type patterns of figure~\ref{fig:patterns}. Rows (c) and (d) share the B--B--b--b pattern and differ by the hole assignment. The b--b--b--b pattern appears twice, for symmetric and asymmetric hole placements, to show that placement rather than the pattern alone fixes the sign there.}
    \label{tab:patterns}
\end{table}
Whether the sign of $\Delta H$ is fixed by the anchor pattern alone depends only on which channels contain bb pairs.
Up to rotation there are six anchor-type patterns on four cyclically ordered anchors (BBBB, 3B+1b, BBbb, BbBb, 3b + 1B, bbbb). 
For every pattern except BBbb the channel content of bb pair is the same for either hole assignment, while BBbb realizes two inequivalent hole assignments appearing in rows (c) and (d) of table~\ref{tab:patterns}.
Row (c) punches the holes on the BB and bb arcs, giving a bb pair in the connected channel only ($\Delta H > 0$) while row (d) punches them on the two Bb arcs, giving a bb pair in the disconnected channel only $(\Delta H < 0)$. 
The bbbb pattern is included twice in table~\ref{tab:patterns} to show that the hole placement, not the pattern, fixes the sign in that case.
Symmetric holes give $\Delta H = 0$ while asymmetric ones give $\Delta H \ne 0$ even though both channels contain the same number of bb pairs.

The B--b--b--B and entanglement arena patterns shown in figure~\ref{fig:geom_examples} share a mechanism: in both, the disconnected channel inherits a short bb geodesic supplying the dominant contribution to $\Delta H$.
b-type anchors do not introduce new phase mechanisms. They shift the threshold of the existing two-interval transition by an amount $\Delta H$ set by the channel imbalance, with the sign of $\Delta H$ fixed by the cyclic anchor pattern when bb pairs occupy a single channel, and by the competition of the $h(\rho_{ij})$ otherwise.

\section{Price and distance for surface/state correspondence}\label{sec:price_and_distance}

In this section, we redefine the price and the distance in terms of pure-state geometries in the sense of the surface/state correspondence \cite{Miyaji:2015yva}.

By the surface/state correspondence, we consider a set of pure states and a Hilbert space on the boundary of a pure-state geometry $\Sigma$. An entanglement wedge $\mW(\tilde{R})$ of a subsystem $\tilde{R}$ on the boundary $\partial \Sigma$ is given by the minimal codimension two surface determined by the $\tilde{RT}$ formula following the correspondence, as stated around \eqref{eq:RT}. Thus, the wedge is the union of the subsystem $\tilde{R}$ and the bulk geodesics anchored on the boundary of the subsystem. We assume geometric complementarity, i.e., $\mW(R^c) = \mW(R)^c$.

Let us now briefly review an operator algebraic quantum error correction\footnote{In this paper, we follow the convention used in \cite{Pastawski:2016qrs} and consider bulk sites in a bulk geometry and deal with type-$I_\infty$ vN algebras $\mA_x$ on each site and $\mA_X$ in a bulk subregion $X$, where $X$ is a collection of bulk sites. For the studies of operator algebraic quantum error correction based on hyperfinite von Neumann algebras, see for example, \cite{Kribs:2006doz, Kribs:2004dqj,Beny:2007rpx,Beny:2007ewj,Kang2021infiniteOAQEC,Beny:2008loq,Harlow:2016vwg,Faulkner:2020hzi,Furuya:2020tzv,Faulkner:2022ada} and the references therein.} in the context of the pure-state geometries. Consider a code subspace $\mH_c = P_c \mH$ where $P_c$ is a projection from a Hilbert space $\mH$ of a quantum system to $\mH_c$. Following \cite{Pastawski:2016qrs}, given the code subspace $\mH_c$ and logical subalgebra  $\mA$, an algebra of logical operators, a subsystem $\tilde{R}$ of $\mH$ is correctable with respect to $\mA$ against the erasure of $\tilde{R}$ if and only if
\begin{equation}
    [PYP,X] = 0
\end{equation}
for all $ X\in \mA$ and every $Y$ supported on $\tilde{R}$. If so, $\mA$ can be reconstructed on the complementary subsystem $\tilde{R}^c$\footnote{These statements are from lemma $1$ and $2$ in \cite{Pastawski:2016qrs}.}. 

In terms of the reconstruction, we assume the wedge reconstruction hypothesis. That is, if the bulk point $x$ is contained in the wedge $\mW(\tilde{R})$ of the boundary subregion $\tilde{R}$, the logical algebra $\mA_x$ at the point $x$ can be reconstructed on $\tilde{R}$.

Based on the wedge reconstruction hypothesis, the distance for $\mA_x$ is the size of the smallest region $\tilde{R}$ which is not correctable with respect to $\mA_x$. The price is the size of the smallest region $\tilde{R}$ which is correctable with respect to $\mA_x$ and can reconstruct $\mA_x$. These notions can be defined for a set $X$ of bulk sites. We now provide the formal definition of the distance and the price for later purposes.









\begin{definition}[Distance for the pure-state geometries]
    The distance for a spacetime point and a spacetime region in a pure-state geometry $\Sigma$ is defined as 
    \begin{equation}
        \tilde{d}(\mA_x) = \underset{\tilde{R}\subseteq \partial \Sigma:x\notin \mW(\tilde{R}^c)}{min} |\tilde{R}|, \; \tilde{d}(\mA_X) =\underset{x\in X}{min} \; \tilde{d}(\mA_x)
    \end{equation}   
    where $\tilde{R}$ is a boundary subregion of the boundary $\partial \Sigma$ of the pure-state geometry $\Sigma$. Moreover, $\mW(\tilde{R}^c)$ is a wedge in the pure-state geometry homologous to boundary subregion $\tilde{R}^c$.
\end{definition}

\begin{definition}[Price for pure-state geometries]
    The price for a spacetime point and a spacetime region in a pure-state geometry
    \begin{equation}
        \tilde{p}(\mA_x) = \underset{\tilde{R}\subseteq \partial \Sigma :x\in \mW(\tilde{R})}{min} |\tilde{R}|, \; \tilde{p}(\mA_X) = \underset{\tilde{R}\subseteq \partial \Sigma:X\subseteq \mW(\tilde{R})}{min} \tilde{p}(\mA_x)
    \end{equation}   
\end{definition}









\section{Uberholography and tighter bounds on price and distance}\label{sec:uberholography}





The main goal of this section is to discuss uberholography in the context of pure-state geometries. 
Uberholography is a property seen in holographic codes on asymptotically negatively curved bulk manifolds, where the local correctability of erasure errors allows the boundary support of a logical operator to be recursively reduced \cite{Almheiri:2014lwa, Pastawski:2015qua, Pastawski:2016qrs}. 
Its main property is that logical information in a bulk subregion can be supported on a boundary region with a fractal structure. 
This suggests that a $d+1$ dimensional bulk can possibly emerge from degrees of freedom supported on a set of fractal dimension $\alpha < d$.
The universality of the fractal dimension has been studied in various contexts \cite{Pastawski:2016qrs, Bao:2022vxc, Bao:2022tgv, Ageev:2022awq, Bhattacharjee:2024ceb}.

We find that the core behavior remains unchanged when looking at the original asymptotic boundary, but that the recursive hole-punching cannot be carried out within the RT-boundary geodesics themselves. 
We will show that the universal fractal dimension on the original asymptotic boundary is unchanged, while the resulting upper bounds on price and distance depend on whether the RT boundaries are retained or traced out before the asymptotic boundary is fractalized.

\subsection{Review of uberholography}\label{sec:uberholography_review}

Here we review the arguments laid out in \cite{Pastawski:2016qrs} that result in the universal fractal dimension $\alpha$. 
Consider a boundary subregion of the original boundary which we label as $R$ with size $|R|$. 
If we punch a hole $H$ in the middle of $R$, how large can $H$ be before the bulk entanglement wedge of $R_1 \cup R_2 := R\setminus H$ is disconnected? 
(Note here that this $H$ is distinct from the $\Delta H$ in section~\ref{sec:phase_dynamics}).
The hole is punched (see figure~\ref{fig:hole_punch}) so that, with kept fraction $0<r<1$ 
\begin{equation}
    \label{eq:hole_punch}
    |H| = (1-r) |R|, \qquad |R_1| = |R_2| = \frac{r}{2}|R|.
\end{equation} 
Now we want to calculate the minimum $r$ (i.e., largest hole) such that the entanglement wedge of $R_1 \cup R_2$ remains connected.\footnote{Here we switch from the disk of section~\ref{sec:phase_dynamics} to the metric on the upper-half plane, such that boundary anchored geodesics have length $L= 2\log\lp|R|/\ep\rp$ where we've set $L_{AdS}=1$ and $\ep$ is a short-distance UV cutoff. The figures of subsection~\ref{sec:purified_uberholography} return to the disk.}
The two candidate surfaces have lengths 
\begin{equation}
    L_\text{conn} = 2\log \frac{|R|}{\ep} + 2\log \frac{(1-r) |R|}{\ep}, \qquad L_\text{disc} = 4 \log \frac{(r/2) |R|}{\ep}.
\end{equation}
The dominant wedge is given by the smaller total RT surface length \cite{Ryu:2006bv, Ryu:2006ef}, so the critical $r$ occurs at $L_\text{conn} = L_\text{disc}$, i.e. $1-r = r^2/4$, which yields 
\begin{equation}
    \label{eq:hole_ratio}
    r = 2 \lp \sqrt{2}-1 \rp.
\end{equation}
The prescription is then to do this hole punching at this critical $r$ a total of $m$ times until we reach some short distance cutoff $\ep$ such that $\ep = (r/2)^m |R|$. 
The amount of remaining boundary after completing this procedure will be 
\begin{equation}
    \label{eq:hole_punched_remaining}
    |R_\text{min}| / \ep = 2^m = (|R|/\ep)^\alpha,
\end{equation}
where 
\begin{equation}
    \label{eq:fractal_dim}
    \alpha = \frac{\log 2}{\log(2/r)} \approx 0.786
\end{equation}
is the universal fractal dimension and $R_\text{min}$ is the resulting union of remaining boundary subregions. 

In this case the universal fractal dimension is defined via \eqref{eq:hole_punched_remaining}. 
In our considerations below, however, $r$ will depend on which recursive step we are on and thus 
$\ep = \prod_{i=1}^m (r_i/2) |R| \ne (r/2)^m |R|$ and so the analysis above breaks down as there is no longer a closed form for $m$ in terms of some constant $r$. We see numerically that the kept fraction converges to the expected value ($r_i \to r_\infty =2(\sqrt{2}-1)$ as $i\to\infty$) in the limit of recursive steps. 
This convergence leads us to define the asymptotic fractal dimension 
\begin{equation}
    \label{eq:fractal_dim_def} 
    \alpha:= \frac{\log2}{\log(2/r_\infty)},
\end{equation}
where $r_\infty$ is the kept fraction (i.e., $|H| = (1-r_\infty) |R|$) in the asymptotic limit of the number of hole punches. 
We will also present numerical evidence (see figure~\ref{fig:hole_punch_ratio}) that in these pure-state geometries the kept fraction $r$ at a given step is upper bounded by \eqref{eq:hole_ratio} 
and thus with this asymptotic definition of $\alpha$ we have 
\begin{equation}
    \label{eq:hole_punched_remaining_new}
    |R_\text{min}| / \ep = 2^m \le (|R|/\ep)^\alpha.
\end{equation}

\begin{figure}[h]
    \centering
    \includegraphics[width=1.\linewidth]{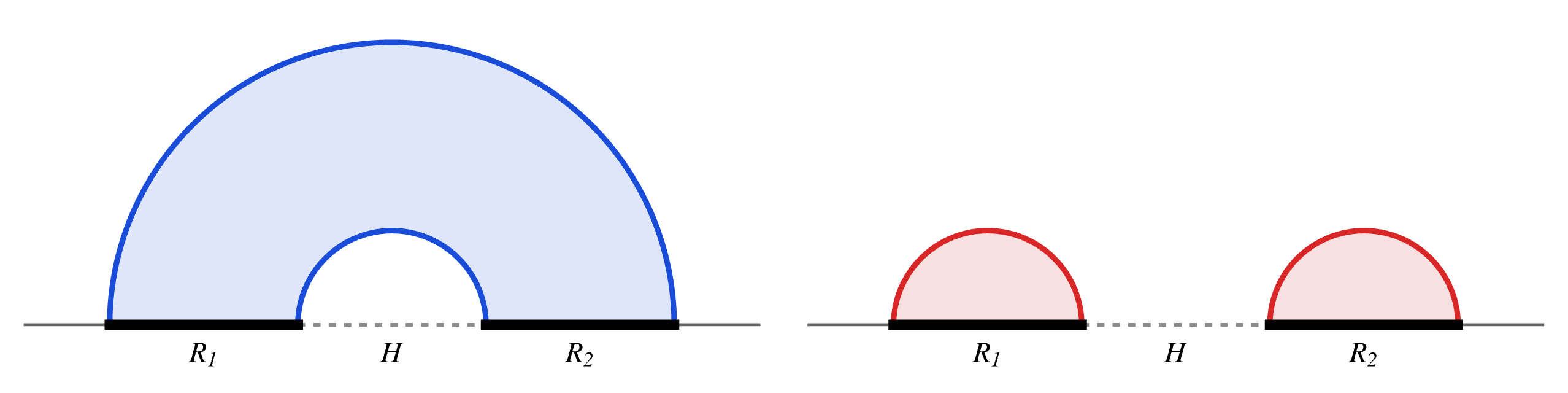}
    \caption{The two competing entanglement wedge candidates when punching hole $H$ into the middle of $R$ with $R_1 \cup R_2 = R \setminus H$. The left is the connected candidate and the right is the disconnected.}
    \label{fig:hole_punch}
\end{figure}

\subsection{Uberholography for pure-state geometries}\label{sec:purified_uberholography}

Uberholography relies on the possibility of two disjoint boundary subregions having a connected entanglement wedge and takes this to the extreme.
However, disjoint subregions lying on a single RT-boundary surface do not have competitive connected wedges and thus cannot exhibit this fractal-like structure. 
As the bulk metric is unaltered throughout the purification, an RT-boundary coincides with a bulk geodesic and the minimal surface homologous to any sub-arc of it is that sub-arc itself. 
Punching a hole contained on a single RT-boundary geodesic therefore yields a replacement geodesic of the same length as the removed segment. 
The connected candidate degenerates onto the boundary, so $L_\text{conn} \ge L_\text{disc}$ for every hole size, no connected wedge competes, and the recursion cannot start. 
See figure~\ref{fig:no_supadd}. 
We emphasize that this is a statement about sub-arcs of a single RT-boundary geodesic. 
A hole that punches out a corner is instead replaced by a strictly shorter corner-cutting geodesic, and holes punched on several RT-boundary geodesics do exhibit genuine phase transitions as shown in section~\ref{sec:examples}.
The difference is that there is no recursion within an RT-boundary geodesic which is required by uberholography. 

\begin{figure}[ht]
    \centering
    \includegraphics[width=1.\linewidth]{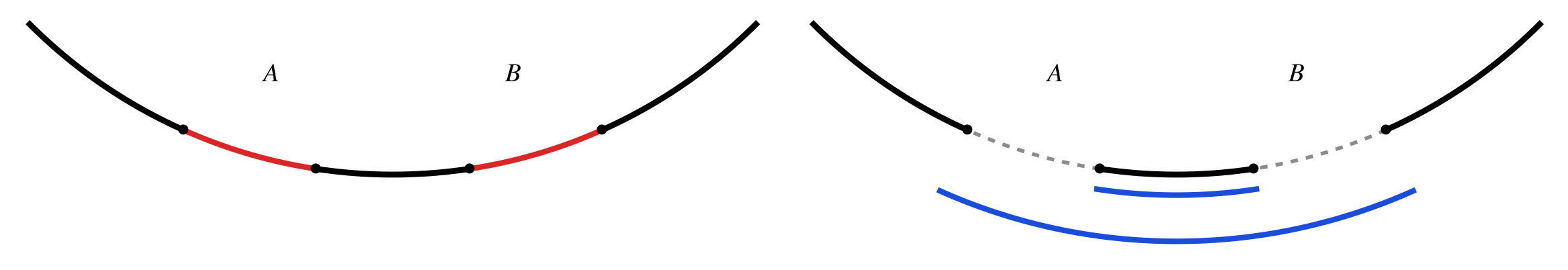}
    \caption{The lack of a competitive connected entanglement wedge along an RT-boundary. 
    Left: the `disconnected' candidate in red.  
    Right: the `connected' candidate in blue.
    Note that by inclusion, the connected surface will always be longer than the disconnected surface. }
    \label{fig:no_supadd}
\end{figure}

While fractal structure is lost on the RT boundaries, any asymptotic boundary left untouched still has this property (as this boundary is not itself a minimal surface in the bulk). 
It is possible, however, that the qualities of this fractal structure are altered by the purification. 

As an example, consider the pure-state geometry with the bulk region $X$ shown in figure~\ref{fig:bulk_subregion} that we would like to determine improved bounds on price and distance for.
The precise configuration, which underlies all quantitative statements in this subsection, is specified in appendix~\ref{app:config}.
As our initial guess of the price we might remove half of the boundary that does not contain any of $X$.
\begin{figure}[ht]
    \centering
    \includegraphics[width=0.5\linewidth]{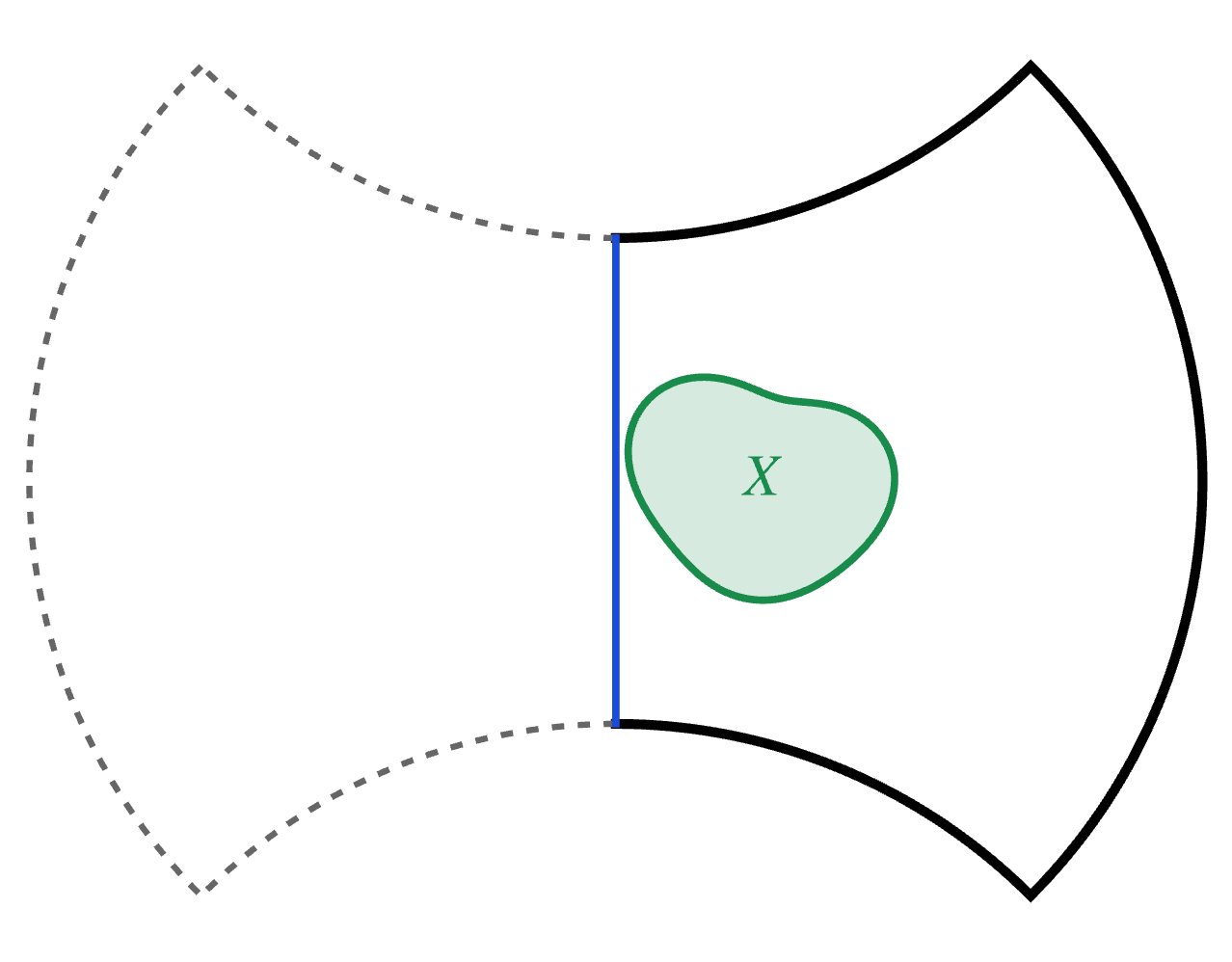}
    \caption{Example of a bulk region $X$ in a  pure-state geometry with an initial guess for the price. The left side of the boundary has been traced out (shown as a dashed line) while the remaining entanglement wedge of the remaining boundary still contains $X$.}
    \label{fig:bulk_subregion}
\end{figure}
This is clearly sub-optimal and as such we can proceed with punching holes on some combination of the asymptotic and RT boundaries. 
The question is then which we should start with. 
As the goal is to find a minimal size boundary subregion that contains $X$ in its entanglement wedge, the size $|\tilde{R}|$ of a boundary subregion is measured as follows.
We count lattice sites, one site per unit proper length (in AdS units) along the regulated subsystem boundary. 
An asymptotic interval of length $\ell$, regulated at $|z| = 1 - \ep$, has proper length $\sim \ell/\ep$ and so contributes $\sim \ell/\ep$ sites, whereas an RT-boundary geodesic of the same span is a bulk geodesic of regulated proper length $\sim\log(\ell/\ep)$ so it contributes only at subleading order compared to the asymptotic interval.
Thus as $\ep\to0$, the asymptotic boundary will dominate the price and distance and so we can focus on the procedure that most minimizes the remaining asymptotic boundary.
Each procedure below exhibits a feasible boundary region whose entanglement wedge contains $X$, hence an upper bound on the price. 
At each step the hole is taken as large as possible subject to two requirements, that the entanglement wedge remains connected and that $X$ remains inside of it.
We assume throughout that $X$ sits deep enough in the bulk that connectivity is the binding requirement, so each hole is taken right at the connected-disconnected phase transition (as in figures~\ref{fig:naive_uberholography} and~\ref{fig:better_uberholography}).
Since the holes shrink as more steps are taken, only the first few steps approach $X$. 

We contrast two possible procedures to emphasize that quantities such as price and distance are procedure dependent while the fractal dimension is not, as explained below.
Both procedures begin from the opening of figure~\ref{fig:bulk_subregion} and follow the same rules.
At step $i$, every remaining asymptotic interval has a new hole punched out at its center with a common kept fraction $r_i$ fixed by the phase transition, and the recursion stops once the remaining intervals reach the lattice spacing $\ep$.

\begin{itemize}
    \item Procedure 1. Holes are punched directly on the asymptotic boundary while the RT-boundary segments are kept throughout (figure~\ref{fig:naive_uberholography}).
    
    \item Procedure 2. The RT-boundary segments are first traced out as far as connectivity allows.
    The standard uberholography of section~\ref{sec:uberholography_review} is then applied to the remaining asymptotic boundary (figure~\ref{fig:better_uberholography}).
\end{itemize}

In Procedure 1 the competing RT-surfaces differ from those in the original uberholography, which gave a constant $r=2( \sqrt{2} - 1)$ for each hole punch. 
The first hole punch, shown in figure~\ref{fig:naive_uberholography}, is precisely the BBbb configuration of example~\ref{ex:bbpp} thus the critical kept fraction $r_1$ follows directly from theorem~\ref{thm:main}.

\begin{figure}[ht]
    \centering
    \includegraphics[width=.9\linewidth]{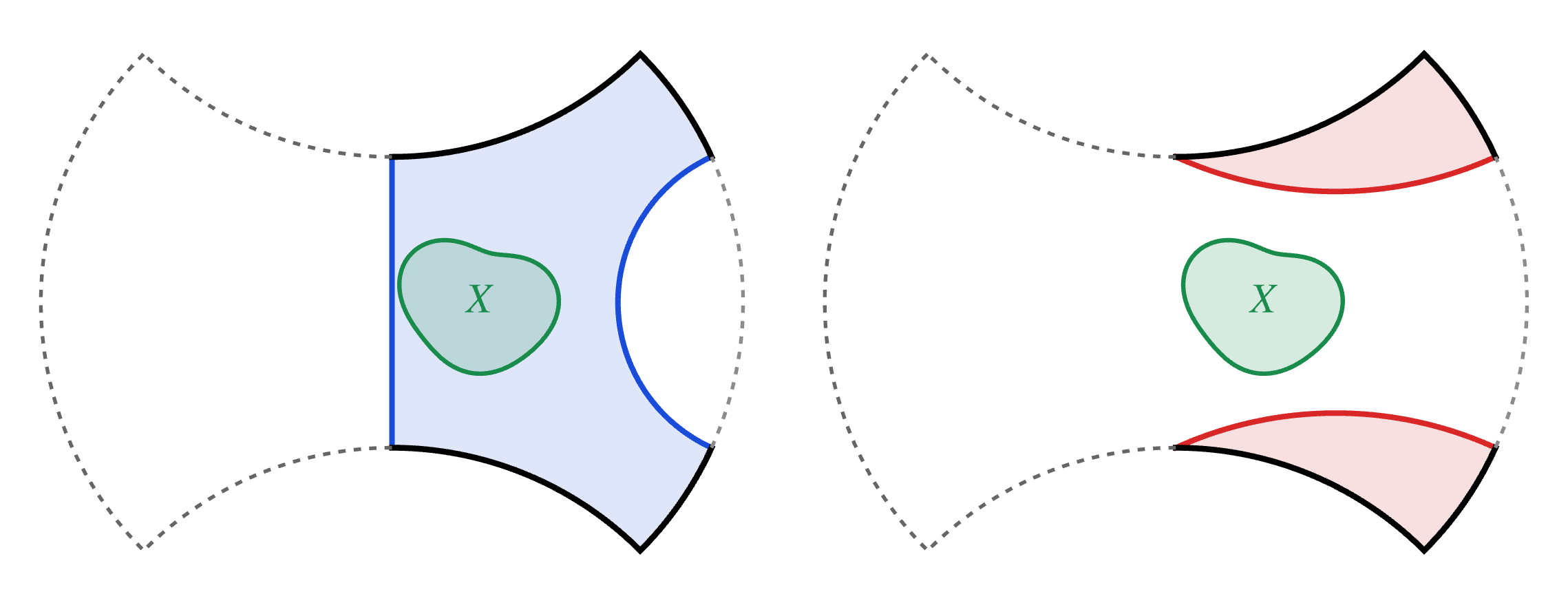}
    \caption{First hole punch for the  procedure starting on the asymptotic boundary. 
    The two competing minimal surfaces differ from the original case and thus the initial $r$ is different.
    The size of the punched region shown is right before/after the transition point for the left/right figures.}
    \label{fig:naive_uberholography}
\end{figure}

In Procedure 2, shown in figure~\ref{fig:better_uberholography}, the competing minimal surfaces become those of the original uberholography treatment and we recover the fractal dimension $\alpha = \frac{\log2}{\log(2/r_\infty)}$ with the constant kept fraction $r_i = r_\infty = 2(\sqrt{2}-1)$ at every step.
This is the reference sequence that saturates \eqref{eq:hole_punched_remaining_new}.

There are a few characteristics to note here. 
In Procedure 1 (figure~\ref{fig:naive_uberholography}) we can see that the initial punched hole is much larger than that in Procedure 2 (figure~\ref{fig:better_uberholography}) corresponding to a smaller first kept fraction $r_1 \approx 0.476$ as compared to $r_\infty = 2(\sqrt{2}-1) \approx 0.828$ (see appendix~\ref{app:config} for the exact value and configuration).
This highlights that the purification provides additional protection of the bulk information against erasure.
The retained RT-boundary geodesics provide additional redundancy, so more of the asymptotic boundary can be erased while still reconstructing $X$.
While this initial punch is much larger, as we proceed iteratively the hole punch fraction $1-r_i$ in Procedure 1 actually decreases (so $r_i$ increases) and asymptotes to the constant $r_\infty=2\lp \sqrt{2} - 1\rp$ (see figure~\ref{fig:hole_punch_ratio}).
Both the bound $r_i \le r_\infty$ and the convergence $r_i \to r_\infty$ are numerical observations in this configuration.
While we do not have an analytic proof, this is what one expects heuristically since later holes are exponentially small and thus far from the b-type anchors so the process locally approaches the original boundary competition of section~\ref{sec:uberholography_review}.
Thus the fractal dimension $\alpha$, with our asymptotic definition, remains the same.

\begin{figure}[ht]
    \centering
    \includegraphics[width=.9\linewidth]{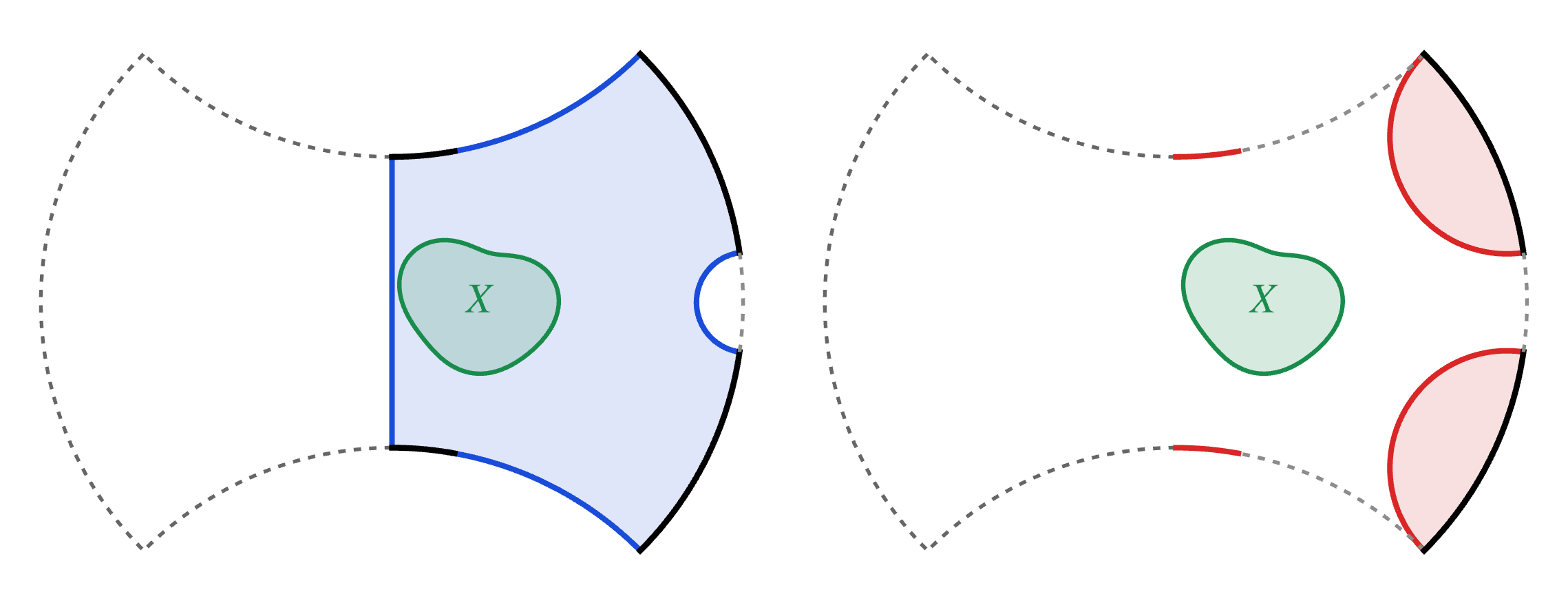}
    \caption{First hole punch for an alternative uberholography procedure.
    The two competing minimal surfaces are now the same as the original case.
    The size of the punched region shown is right before/after the transition point for the left/right figures.
    The RT-boundary segments have been opened to the critical opening fraction of appendix~\ref{app:config}, leaving the short RT-boundary remainders visible near $p_\pm$.}
    \label{fig:better_uberholography}
\end{figure}

\begin{figure}[ht]
    \centering
    \includegraphics[width=.7\linewidth]{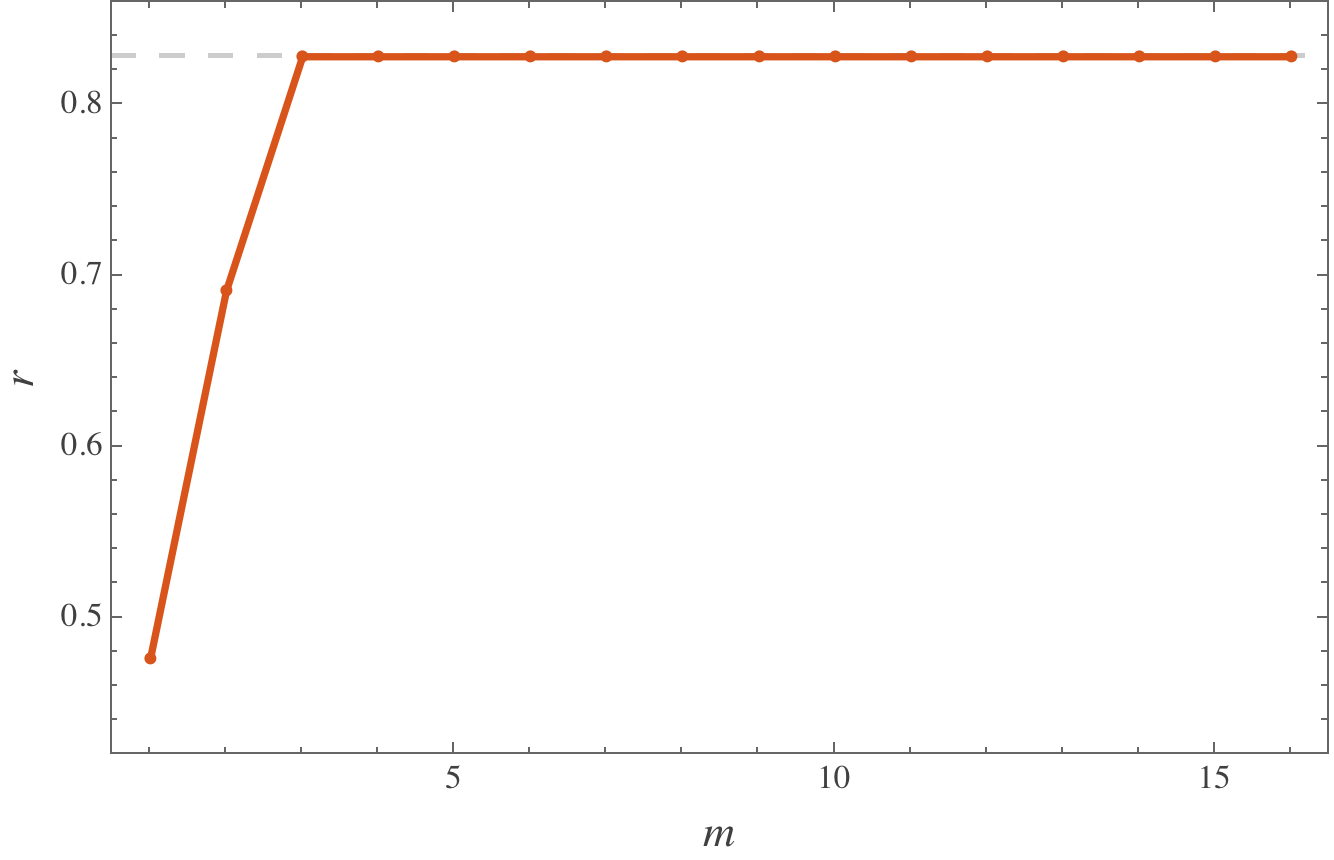}
    \caption{A plot of the ratio of kept region $r$ with respect to step number $m$ for Procedure 1. The grey dashed line corresponds to $r_\infty = 2(\sqrt{2}-1)$, which the ratio attains from the third step on. All sixteen steps entering figure~\ref{fig:total_remaining_boundary} are shown, and the values are tabulated in table~\ref{tab:kept_fractions} of appendix~\ref{app:config}. }
    \label{fig:hole_punch_ratio}
\end{figure}

The two procedures need not give the same bound on the price and distance, and the comparison is fixed by
\eqref{eq:hole_punched_remaining_new}.
Procedure 1 keeps a smaller fraction $r_i \le r_\infty$ at each step and reaches the cutoff in no more steps, so it never leaves more asymptotic boundary than Procedure 2.
Both procedures asymptote to the same kept fraction $r_\infty$ and so share the fractal dimension $\alpha$ of \eqref{eq:fractal_dim_def}.
The remaining question then is how much tighter Procedure 1 is.
Procedure 2 punches holes at the constant critical rate and reproduces $|R_\text{min}|/\ep = (|R|/\ep)^\alpha$, while Procedure 1 removes more at each early step.
This comparison can be made precise, in the form of a proposition whose hypotheses are the numerically observed properties of the kept fractions.

\begin{proposition}\label{prop:comparison}
    Suppose the kept fractions $r_i$ of Procedure 1 satisfy
    (i) $r_i \le r_\infty$ for all $i$,
    (ii) $r_1 < r_\infty^2$, and
    (iii) the product $\prod_{i= 1}^{m_1} (r_\infty/r_i)< 2/r_\infty$.
    For every cutoff $\ep$ small enough that $m_1 \ge 1$ below, Procedure 1 leaves strictly less asymptotic boundary than Procedure 2,
    \begin{equation}
        |R_\text{min}^{(1)}| < |R_\text{min}^{(2)}|.
    \end{equation}
\end{proposition}

Properties (i)–(iii) are numerical observations for the configuration shown in appendix \ref{app:prop_assumptions}, for which $r_1 \approx 0.476 < r_\infty^2 \approx 0.686$.
The sequence $r_i$ is a property of the recursion and does not depend on the cutoff $\ep$, which enters only through the stopping step $m_1$, so no limit is involved in (iii).
Only the first two factors of the product differ from one (table~\ref{tab:kept_fractions}), and so $\prod_{i= 1}^{m_1} (r_\infty/r_i)=(r_\infty/r_1)(r_\infty/r_2) \approx 2.08 < 2/r_\infty = 1 + 1/\sqrt{2} \approx 2.41$.
We do not have an analytic proof of (i)--(iii).

\begin{proof}
Fix $\ep > 0$ and let $m_1, m_2$ be the last steps at which the intervals of Procedures 1 and 2 remain at or above the cutoff, so that
\begin{equation}\label{eq:ep1_ep2}
    \ep \le \ep_1 := \prod_{i=1}^{m_1} \frac{r_i}{2} |R|,
    \qquad
    \ep \le \ep_2 := \lp\frac{r_\infty}{2}\rp^{m_2} |R|,
\end{equation}
while one further step would drop below $\ep$.
The surviving boundary consists of $2^{m_k}$ intervals of size $\ep_k$ (where $k\in\{1, 2\}$), so
\begin{equation}
    |R_\text{min}^{(1)}| = 2^{m_1}\ep_1 = \prod_{i=1}^{m_1} r_i |R|,
    \qquad
    |R_\text{min}^{(2)}| = 2^{m_2}\ep_2 = r_\infty^{m_2} |R|.
\end{equation}
First, by hypothesis (iii),
\begin{equation}
    \lp \frac{r_\infty}{2}\rp^{m_1+2} |R| < \prod_{i=1}^{m_1+1} \frac{r_i}{2}|R| < \ep,
\end{equation}
where the last inequality holds by definition of $m_1$.
Procedure 2 therefore stops by step $m_1 + 1$ so $m_2 \le m_1 + 1$.
Second, by hypotheses (i) and (ii),
\begin{equation}
    \prod_{i=1}^{m_1} r_i < r_\infty^{m_1+1} \le r_\infty^{m_2},
\end{equation}
where the last step uses $m_2 \le m_1 + 1$ and $r_\infty < 1$.
Therefore
\begin{equation}
    |R^{(1)}_\text{min}| = \prod_{i=1}^{m_1} r_i |R| < r_\infty^{m_2}|R| = |R_\text{min}^{(2)}|.
\end{equation}
\end{proof}

Given the hypotheses, Procedure 1 thus leaves strictly less asymptotic boundary at every sufficiently small cutoff, so its upper bounds on the price and distance are strictly tighter than Procedure 2's. Figure~\ref{fig:total_remaining_boundary} shows this, with Procedure 1's remaining boundary lying uniformly below Procedure 2's.
The RT-boundary retained in Procedure 1 adds only a subleading $\sim\log(\ell/\ep)$ to $|\tilde{R}|$ and does not change the comparison.

\begin{figure}[ht]
    \centering
    \includegraphics[width=.7\linewidth]{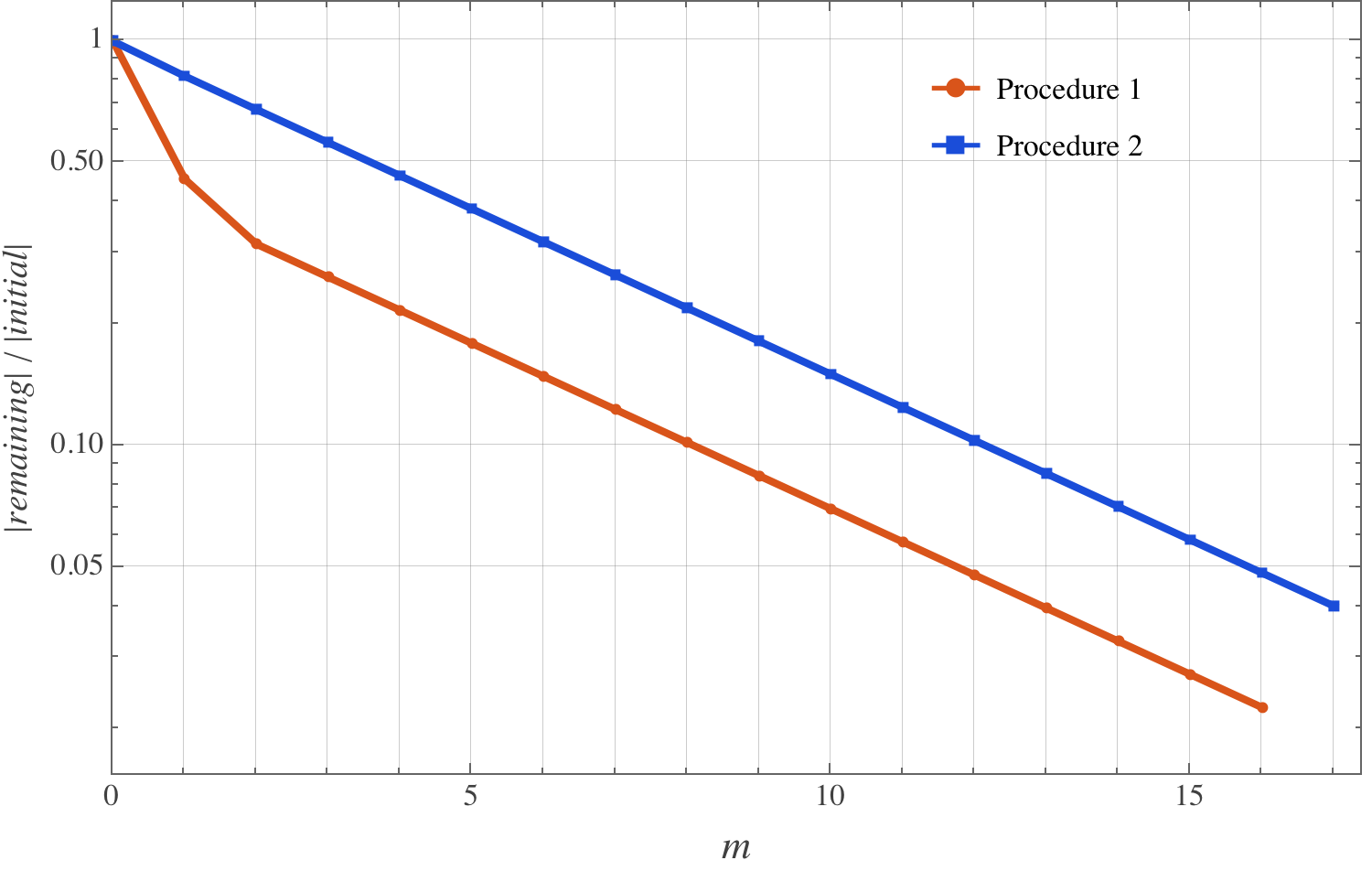}
    \caption{
    A plot of the total remaining asymptotic boundary (as a ratio of the initial boundary size) with respect to step number $m$ for both procedures with $\ep = 10^{-6}$ in units of arclength along the unit circle.
    In this configuration, Procedure 1's intervals first drop below the cutoff at its 16th step while Procedure 2's do so at its 17th, corresponding to $m_1 = 16$ and $m_2 = 17$ in the convention of \eqref{eq:ep1_ep2}. See appendix~\ref{app:config}.}
    \label{fig:total_remaining_boundary}
\end{figure}
As previously discussed, because the asymptotic boundary will dominate the price as $\ep\to0$, this implies that fractalizing the asymptotic boundaries before considering the RT boundaries yields tighter upper bounds on price and distance.
The intuition here is that the RT boundaries have already been optimized and distilled \cite{Bao:2023til} while the asymptotic boundary has not.
We do not have a proof of this expectation and leave it as an open question.

\section{Discussion}\label{sec:discussion}

We have studied the entanglement wedge phase structure and the code properties of pure-state holographic geometries in which boundary subregions have been replaced by their RT geodesics via the surface/state correspondence.
Our results fall into two parts.

First, in section~\ref{sec:phase_dynamics} we reduced the connected/disconnected phase transition of two holes to a single cross-ratio threshold relation, $\eta'/\eta = e^{\Delta H/2}$, in which the standard two-interval threshold is shifted by $\Delta H$ which is determined by geodesics anchored only on the RT-boundary.

Second, in sections~\ref{sec:price_and_distance} and \ref{sec:uberholography} we used this phase structure to study price, distance, and uberholography.
The recursion underlying uberholography cannot run within an RT-boundary geodesic, since such a boundary is already a bulk minimal surface and admits no competitive connected wedge.
We find numerically that the universal fractal dimension $\alpha \approx 0.786$ of the asymptotic boundary is unaffected by the purification, while the price and distance depend on whether the RT boundaries are retained or traced out before the asymptotic boundary is fractalized.
In the configurations we study, fractalizing the asymptotic boundary while retaining the RT boundaries gives tighter bounds, consistent with the intuition, which remains unproven, that the RT-boundary has already been distilled.
That $\alpha$ survives the purification while coarse quantities such as the price and distance do not is a pattern that aligns with the behavior found for fractal boundary erasures under bulk deformations studied in \cite{Bao:2022vxc, Bhattacharjee:2024ceb}.

Several possible future directions remain open. Throughout we assumed exact entanglement wedge reconstruction which can be relaxed to approximate holographic codes \cite{Furuya:2021lgx}.
This would connect the Singleton gap $\delta_X = \tilde p(\mA_X) - \tilde d(\mA_X) - \tilde{k}_X$ to the breakdown of complementary recovery.
In the operator algebraic language, complementary recovery is closely tied to Haag duality for the boundary subregion algebras, and it is natural to ask whether a nonvanishing Singleton gap for the code of a pure-state geometry can be traced to a controlled failure of Haag duality, with the gap playing a role analogous to a Jones index for the inclusion of reconstructible algebras.
We leave a precise formulation of this correspondence to future work.
It is also of interest to extend the cross-ratio threshold relation beyond the bipartite setup considered here and to ask whether there is a quantity like $\Delta H$ in multipartite contexts, for instance for the multipartite entanglement of purification \cite{Umemoto:2018jpc, Bao:2018gck}.
During the completion of this work, a related paper introducing the entanglement wedge polygon (EWP) appeared~\cite{Fujiki:2026ucr}. The EWP is a codimension-one volume bounded by RT surfaces for pure states and entanglement wedge cross sections for mixed states. It is very closely related to the entanglement arena studied here and offers a complementary, volume-based perspective on these holographic geometries.

\section*{Acknowledgments}
We would like to thank Joydeep Naskar for discussions. N.\,B. is supported by the DOE Office of Science-ASCR, in particular the grand Novel Quantum Algorithms from Fast Classical Transforms and by Northeastern University. K.\,F. acknowledges support from Professor Ning Bao at Northeastern University. J.\,M is supported by a graduate assistantship from Northeastern University.

\appendix

\section{Numerical configuration for section~\ref{sec:uberholography}}\label{app:config}

Here we specify the configuration underlying the quantitative statements of section~\ref{sec:purified_uberholography}, so that figures~\ref{fig:hole_punch_ratio} and \ref{fig:total_remaining_boundary} and the quoted kept fractions can be reproduced.

\subsection{Setup}
The pure-state geometry is symmetric under reflection across both axes of the disk.
Its corners sit at $e^{\pm i\pi/4}$ and $e^{\pm 3i\pi/4}$.
The original asymptotic region is the right-hand arc $R = \{e^{i\phi} : \phi \in (-\pi/4, \pi/4)\}$, of arclength $|R| = \pi/2$, together with its mirror image on the left.
The RT-boundary geodesics are the two bulk geodesics connecting $e^{i\pi/4}$ to $e^{3i\pi/4}$ (the unit-radius circle centered at $i\sqrt{2}$) and $e^{-i\pi/4}$ to $e^{-3i\pi/4}$ (its reflection).
The entanglement wedge cross section between them is the segment of the imaginary axis with feet
\begin{equation}
    p_\pm = \pm i \lp\sqrt{2}-1\rp,
\end{equation}
for which $\rho(p_+, p_-) = 1$ and the EWCS length is $\arccosh(3) = 2\log(1+\sqrt{2})$.

The initial opening of figure~\ref{fig:bulk_subregion} traces out the left asymptotic arc together with the segments of both RT-boundary geodesics beyond the EWCS feet.
The retained subsystem is therefore $R$ together with the RT-boundary segments running from $e^{\pm i\pi/4}$ to $p_\pm$, and the bulk region $X$ lies inside the wedge bounded by $R$, these segments, and the EWCS.
Procedure 2 opens the RT-boundary segments further, up to the critical fraction beyond which the entanglement wedge would disconnect, so its wedge remains bounded by the EWCS and feasibility only requires $X$ to lie on the same side of the EWCS as $R$.
The remaining placement constraint is the standing assumption that $X$ avoids every hole cap, the deepest of which is the first bite of Procedure 1, which at criticality reaches $x = (1-\sin\theta_1)/\cos\theta_1 \approx 0.643$ on the real axis.
The region drawn in figures~\ref{fig:bulk_subregion}, \ref{fig:naive_uberholography}, and \ref{fig:better_uberholography} satisfies both requirements.

The first hole of Procedure 1 is punched at the center of $R$, with regulated B-type anchors $e^{\pm i\theta}(1-\ep)$ and b-type anchors $p_\pm$.
This is the B--B--b--b configuration of example~\ref{ex:bbpp} with $\Delta H = h(1) = 2\log(1+\sqrt2) - \log 4$, so theorem~\ref{thm:main} places the transition at $\eta'/\eta = (1+\sqrt{2})/2$.
The threshold condition reduces to $\sin\theta_1 = \sqrt{2}-1$, giving the critical half-angle in closed form,
\begin{equation}
    \theta_1 = \arcsin\lp\sqrt{2}-1\rp \approx 0.427079
\end{equation}
The kept fraction is measured in the upper-half-plane coordinate $w = -i(z-1)/(z+1)$, which acts on the boundary as $\phi \mapsto \tan(\phi/2)$ and maps $R$ to the real interval $(-\tan(\pi/8), \tan(\pi/8))$.
In this coordinate
\begin{equation}
    r_1 = 1 - \frac{\tan(\theta_1/2)}{\tan(\pi/8)} \approx 0.476489,
\end{equation}
which is the value quoted in section~\ref{sec:purified_uberholography}.
At subsequent steps every remaining interval of $R$ is punched at its center in the $w$ coordinate with the common kept fraction $r_i$ fixed by the dynamic-programming competition described in section~\ref{sec:purified_uberholography}, in which the disconnected length is minimized exactly over all non-crossing pairings of the anchors.
The resulting sequence $r_i$ increases toward $r_\infty = 2(\sqrt{2}-1)$ as shown in figure~\ref{fig:hole_punch_ratio} and tabulated in table~\ref{tab:kept_fractions}.
In both procedures the per-step solve is replaced by the fixed-point value once the solved bite fraction is within $10^{-5}$ of it, which happens after four steps for Procedure 1. 

For Procedure 2 the RT-boundary segments are opened from the corners toward the EWCS up to the critical opening fraction (found to be $\approx 0.771887$), itself solved from the phase transition of the opening, and the kept fraction is then solved step by step with the disconnected competitor taken to be the homology-constrained adjacent pairing of anchors.
The solved fractions coincide with $r_\infty$ to the recursion's $10^{-5}$ tolerance, so the Procedure 2 curve of figure~\ref{fig:total_remaining_boundary} follows the constant-$r_\infty$ sequence used in the text.
The short RT-boundary remainders left by the opening are visible near $p_\pm$ in figure~\ref{fig:better_uberholography}.

\subsection{Results for proposition~\ref{prop:comparison} assumptions}\label{app:prop_assumptions}

\begin{table}[ht]
    \centering
    \begin{tabular}{c l c}
        \hline
        $i$ & \multicolumn{1}{c}{$r_i$} \\
        \hline
        1 & 0.47648905  \\
        2 & 0.69178000  \\
        3 & 0.82842712  \\
        4 & 0.82842712  \\
        5--16 & 0.82842712  \\
        \hline
    \end{tabular}
    \caption{Kept fractions $r_i$ of Procedure 1 in the upper-half-plane for the sixteen steps entering figures~\ref{fig:hole_punch_ratio} and~\ref{fig:total_remaining_boundary}.
    The first four kept fractions were solved via direct comparison of all competing minimal surfaces. 
    Then, in order to compute the number of steps $m_1$ carried out before the $\ep$-cutoff, we fixed $r_i = r_\infty$ for $i\ge 5$ as the direct computation of the competing surfaces became too expensive.
    We find that the step-17 intervals fall below the cutoff $\ep$, giving $m_1 = 16$ in the convention of \eqref{eq:ep1_ep2}.}
    \label{tab:kept_fractions}
\end{table}

Here we show that assumptions (i)-(iii) are satisfied numerically in the setup described above. 
The kept fraction results at each step of Procedure 1 are shown in table~\ref{tab:kept_fractions}. 
The first step is all that is required for satisfying proposition~\ref{prop:comparison} assumption (ii) as $r_1 \approx 0.476 < r_\infty^2 \approx 0.686$.
For assumptions (i) and (iii) we need the additional well supported assumption that the competing connected and disconnected minimal surfaces asymptote to those considered in \cite{Pastawski:2016qrs}.
This has support from the first four steps in this table which are directly computed and show both $r_i \le r_\infty$ and $r_i \to r_\infty$ quickly (by the third step).
Thus with the assumption that this asymptote holds for further steps, we have  both $r_i\le r_\infty$ (assumption (i)) and $\prod_{i_1}^{m_1} r_\infty/r_i \approx 2.08 < 2/r_\infty \approx 2.41$ (assumption (iii)).

\bibliographystyle{JHEP}
\bibliography{main}

\end{document}